\documentclass[a4paper,UKenglish,cleveref, autoref, thm-restate]{lipics-v2019}
\usepackage[utf8]{inputenc}%for certain Polish characters
\usepackage[T1]{fontenc}%for certain Polish characters
\usepackage{amsmath,amssymb,amsfonts,amsthm}
\usepackage{MnSymbol}
\usepackage{xspace}
\usepackage{todonotes}
\usepackage{graphicx}

\usepackage{xcolor}
\usepackage{mathtools}
\usepackage[inline]{enumitem}

\usepackage{environ}

\newtheoremstyle{plainrep}{}{}
  {\itshape}{}
  {\bfseries}{.}{ }%
  {\thmnote{{#3}}}

\theoremstyle{plainrep}
\newtheorem{repetition}{}[section]
\theoremstyle{plain}

\makeatletter
\NewEnviron{recallable}[1]{%
	\edef\recallable@temp{%
		\unexpanded{\expandafter\gdef\csname recallable@#1\endcsname}{\unexpanded\expandafter{\BODY}}%
	}%
	\recallable@temp%
	\BODY%
}
\newcommand{\recall}[1]{
	\renewcommand{\label}[1]{}%
	\csname recallable@#1\endcsname%
}
\newcommand{\recallthm}[1]{
	\begin{repetition}[\Cref{#1} (Restatement)]%
	\recall{#1}%
	\end{repetition} %
}
\makeatother

\newcommand{\NP}{\mathrm{NP}}
\newcommand{\classP}{\mathrm{P}}
\newcommand{\TSPN}{\textsf{TSPN}\xspace}
\newcommand{\TSP}{\textsf{TSP}\xspace}
\newcommand{\GST}{\textsf{GST}\xspace}
\newcommand{\STGST}{\textsf{STGST}\xspace}

\DeclareMathOperator{\E}{E}
\DeclareMathOperator{\Prob}{Pr}
\DeclareMathOperator{\OPT}{OPT}
\DeclareMathOperator{\height}{height}

\newcommand{\Reals}{\mathbb{R}}
\newcommand{\dist}{\mathrm{dist}}
\newcommand{\cL}{\mathcal{L}}
\newcommand{\cC}{\mathcal{C}}
\newcommand{\cY}{\mathcal{Y}}
\newcommand{\an}{\sphericalangle}
\newcommand{\cost}{cost}
\newcommand{\etal}{\emph{et~al.}}
\newcommand{\eps}{\varepsilon}

\DeclarePairedDelimiter\set{\{}{\}}
\DeclarePairedDelimiter\paren{(}{)}
\DeclarePairedDelimiter\brac{[}{]}

\title{On the Approximability of the Traveling Salesman Problem with Line 
  Neighborhoods}

\titlerunning{On the Approximability of TSP with Line 
  Neighborhoods} %optional, please use if title is longer than one line

\author{Antonios Antoniadis}{University of Cologne, Germany}{antoniadis@cs.uni-koeln.de}{}{Work done in part while the
    author was at Saarland University and Max Planck Institute for Informatics and supported by DFG grant AN 1262/1-1.}
\author{S\'andor Kisfaludi-Bak}{Max Planck Institute for Informatics, Saarbr\"ucken, Germany}{sandor.kisfaludi-bak@mpi-inf.mpg.de}{}{}
\author{Bundit Laekhanukit}{Shanghai University of Finance and Economics, Shanghai, China}{bundit@sufe.edu.cn}{}{}
\author{Daniel Vaz}{Operations Research Group, TU Munich, Germany}{daniel.vaz@tum.de}{}{This work has
    been supported by the Alexander von Humboldt Foundation with funds
    from the German Federal Ministry of Education and Research
    (BMBF). Work done in part while the author was at Saarland
    University and Max Planck Institute for Informatics.}

\authorrunning{A. Antoniadis, S. Kisfaludi-Bak, B. Laekhanukit, D. Vaz} %TODO mandatory. First: Use abbreviated first/middle names. Second (only in severe cases): Use first author plus 'et al.'

\Copyright{Antonios Antoniadis, S\'andor Kisfaludi-Bak, Bundit Laekhanukit, and Daniel Vaz} %TODO mandatory, please use full first names. LIPIcs license is "CC-BY";  http://creativecommons.org/licenses/by/3.0/

%TODO mandatory: Please choose ACM 2012 classifications from https://dl.acm.org/ccs/ccs_flat.cfm 

 %TODO mandatory; please add comma-separated list of keywords

\category{} %optional, e.g. invited paper

\relatedversion{} %optional, e.g. full version hosted on arXiv, HAL, or other respository/website
\supplement{}%optional, e.g. related research data, source code, ... hosted on a repository like zenodo, figshare, GitHub, ...

\acknowledgements{}%{I want to thank \dots}%optional

\nolinenumbers %uncomment to disable line numbering

\hideLIPIcs  %uncomment to remove references to LIPIcs series (logo, DOI, ...), e.g. when preparing a pre-final version to be uploaded to arXiv or another public repository

\EventEditors{John Q. Open and Joan R. Access}
\EventNoEds{2}
\EventLongTitle{42nd Conference on Very Important Topics (CVIT 2016)}
\EventShortTitle{CVIT 2016}
\EventAcronym{CVIT}
\EventYear{2016}
\EventDate{December 24--27, 2016}
\EventLocation{Little Whinging, United Kingdom}
\EventLogo{}
\SeriesVolume{42}
\ArticleNo{23}
\begin{document}
\maketitle
\nolinenumbers
\begin{abstract}
We study the variant of the Euclidean Traveling Salesman problem where
instead of a set of points, we are given a set of  lines as input, and the goal is to find the shortest tour 
that visits each line. The best known upper and lower bounds for the
problem in $\Reals^d$, with $d\ge 3$,
are $\NP$-hardness and an $O(\log^3 n)$-approximation algorithm which
is based on a reduction to the group Steiner tree problem. 

We show that TSP with lines in $\Reals^d$ is APX-hard for any $d\ge 3$. More generally, this implies that TSP with $k$-dimensional flats does not admit a PTAS for any $1\le k \leq d-2$ unless $\classP=\NP$, which gives a complete classification of the approximability of these problems, as there are known PTASes for $k=0$ (i.e., points) and $k=d-1$ (hyperplanes).
We are able to give a stronger inapproximability factor for $d=O(\log n)$ by
showing that TSP with lines does not admit a $(2-\eps)$-approximation
in $d$ dimensions under the Unique Games Conjecture. 
On the positive side, we leverage recent results on
restricted variants of the group Steiner tree problem in order to give
an $O(\log^2 n)$-approximation algorithm for the problem, albeit with
a running time of $n^{O(\log\log n)}$.

%The previous state of the art was $\NP$-hardness for $d\ge 3$ and a
%polynomial-time $O(\log^3 n)$-approximation algorithm by a reduction to
%the group Steiner tree problem.
\end{abstract}

%!TEX root=./paper.tex

\section{Introduction}

In the Euclidean Traveling Salesman problem, one is given $n$ points
in $d$-dimensional Euclidean space (denoted by $\Reals^d$), and the
goal is to find the shortest tour visiting all the points. The problem
is $\NP$-hard for $d\geq 2$~\cite{Papadimitriou77}, but it has a
celebrated polynomial time approximation
scheme (PTAS), i.e., a polynomial-time algorithm that produces a tour of length at most
    $(1+\eps)$ times the optimum for any fixed $\eps>0$, due to
  Arora~\cite{Arora98} and (independently) by
  Mitchell~\cite{Mitchell99}. The running time was later improved by
  Rao and Smith~\cite{RaoSmith98}. %These algorithms are all
                                %polynomial time approximation schemes
                                %(PTASes): they give a tour of length
                                %at most $(1+\eps)$ times the optimum
                                %for any fixed $\epsilon>0$ in
                                %polynomial time.

In the past decades, a {considerable amount} of work has concentrated on finding
approximations for variants and generalizations of the Euclidean
Traveling Salesman Problem, e.g., by changing the underlying space~\cite{AroraGKKW98,KrauthgamerL06,DemaineHK11,BartalGK16}, or the objects being visited~\cite{TSPNvaryingSizeAPXhardness, CorridorConnectionProblem,
	TSPNdoublingMetrics, DumitrescuM03, sodaFatRegions,
	Mitchell2010-pairwisedisjoint-connected, Jia2019geometric}. In the latter case
      which is known as the \emph{Traveling Salesman Problem with
      Neighborhoods (TSPN)}, the input consists of $n$
      \emph{neighborhoods}, and the goal is to find the shortest tour
      that visits each neighborhood. More formally, we are given the
      sets $S_1,\dots,S_n\subset \Reals^d$, and we wish to compute the
      shortest closed curve $\tau$ such that for each
      $i\in\{1,\dots,n\}$ we have $S_i \cap \tau \neq \emptyset$. (Observe that the optimum curve $\tau$ consists of
      at most $n$ segments.) %One can define several variants depending
%      on the type of neighborhoods $S_i$ allowed; see also our related
 %     work section below. 

In contrast to regular TSP, \TSPN is already
APX-hard in the Euclidean plane~\cite{TSPNvaryingSizeAPXhardness}, 
i.e., it has no PTAS unless $\classP = \NP$.
Worse still, even the basic case in which each neighborhood is an arbitrary finite set of points in the
Euclidean plane (the so called \emph{Group TSP}) admits no
polynomial-time~${O(1)}$-approximation (unless P~${=}$
NP)~\cite{TSPNcomplexity}.
Even in the case in which each
neighborhood consists of exactly two points~\cite{TSPNpointpairAPX} 
the problem remains APX-hard. 

This inherent hardness of \TSPN gives rise to studying variants of the
problem in which the neighborhoods are restricted in some ways. In a
seminal paper, Arkin and Hassin~\cite{ArkinHassinTSPN} looked into the
problem for various cases of \emph{bounded neighborhoods}, including
translates of convex regions and parallel unit segments, and gave
constant-factor approximation algorithms for them. The best known
approximation algorithm for a more general case of bounded neighborhoods
in the plane is due to Mata and
Mitchell~\cite{DBLP:conf/compgeom/MataM97} and attains an~${O(\log
  n)}$ approximation factor. However, there exist special cases of such bounded
neighborhoods in the plane that do allow for $O(1)$-approximation
algorithms. These include neighborhoods which are disjoint, fat, or
have comparable sizes~\cite{TSPNvaryingSizeAPXhardness, CorridorConnectionProblem,
	TSPNdoublingMetrics, DumitrescuM03, sodaFatRegions,
	Mitchell2010-pairwisedisjoint-connected}.

The complementary case of \TSPN in which \emph{neighborhoods} are \emph{unbounded} regions (which is also the
focus of this paper) is, in general, less well understood. Consider neighborhoods
      that are affine subspaces (\emph{flats}) of dimension $k<d$ in
      $\Reals^d$.  On the positive side, and despite the APX-hardness of the general \TSPN problem already
      in $\Reals^2$, the version with flats (in this case lines) as
      neighborhoods can be solved exactly in $O(n^4\log n)$-time via a
      reduction to the shortest watchman route
      problem~\cite{Jonsson02,watchman03}.  Furthermore, Dumitrescu~\cite{TSPNlinesRays} provides
      a~${1.28}$-approximation algorithm that runs in linear time. In $\Reals^3$, the problem
      of line and plane neighborhoods was first raised by Dumitrescu
      and Mitchell~\cite{DumitrescuM03}. For the line case, they
      already point out that the problem is $\NP$-hard as a direct
      consequence of the $\NP$-hardness of Euclidean TSP in the
      plane~\cite{Papadimitriou77}.  
      Although this leaves the possibility for a PTAS open, the best
      known approximation algorithm  to date for \TSPN with line neighborhoods in $\Reals^3$ was given by
      Dumitrescu and T\'oth~\cite{DumitrescuT16} and achieves an
      $O(\log^3 n)$-approximation.   For the case of
      $(d-1)$-dimensional flats in $\Reals^d$ (which also
      includes planes in
      $\Reals^3$), they give a linear-time (for any constant dimension~$d$ and any
constant~${\eps>0}$) ${(1+\eps)2^{d-1}/\sqrt{d}}$-approximation. This result was subsequently
improved by Antoniadis~\etal~\cite{Antoniadis0HS19} to an EPTAS that
also runs in linear time for fixed $d$ and $\eps$. Whether this variant is $\NP$-hard or not
      remains an interesting open problem. As for the case of
      line neighborhoods in $\Reals^d$ for $d\ge 3$, a PTAS for
      $k$-dimensional flats for $1\le k \leq d-2$ also remained out of reach.

We show that unless $\classP=\NP$, there is no PTAS
for lines in $\Reals^3$. As a direct consequence, we can rule out the
existence of a PTAS in all remaining open cases of \TSPN with flats: there
is no PTAS for $k$-dimensional flat neighborhoods for any $1\leq k
\leq d-2$, unless $\classP=\NP$. 

Let us call the Euclidean TSP problem in $\Reals^d$ with $k$-dimensional flat neighborhoods $(k,d)$-\TSPN.
Although ruling out a PTAS for $(1,3)$-\TSPN is an important step
towards settling the approximability of the problem, the
inapproximability factor obtained is very close to $1$. It would be
desirable to obtain a stronger inapproximability factor, especially given how far we are from any
constant-approximation algorithm for the problem. A natural way to
obtain such a stronger inapproximability result is to consider the
problem in higher dimensional spaces. For example, regarding the
classic Euclidean TSP, it is known that the problem becomes APX-hard for $d=\log
n$~\cite{Trevisan00}. This result directly implies that \TSPN  with
line neighborhoods in $\Reals^{1+\log n}$ is APX-hard, but this is
barely satisfactory, since it again only gives a small
inapproximability factor. However, by using a different reduction from
the vertex cover problem, we are able to show that the problem has no polynomial $(2-\eps)$-approximation in $\Reals^{O(\log n)}$ for any fixed $\eps>0$ under the Unique Games Conjecture~\cite{Khot02a}.

On the algorithmic side, very little is known about
$(k,d)$-\TSPN. For $d=3$, the best known polynomial time
approximation for $(1,3)$-\TSPN is the aforementioned $O(\log^3
n)$-approximation algorithm due to Dumitrescu and
T\'oth~\cite{DumitrescuT16}. Their approach is to discretize the problem by
selecting a polynomial number of ``relevant'' points on each line. It
is shown that restricting the solution to visiting lines at these
points only increases the tour length by a constant factor. The
resulting instance can now be seen as an instance of
group-\TSP, where the relevant points of each line form a
group. By feeding this into the $O(\log^3 n)$-approximation algorithm for general group
Steiner tree~\cite{GargKR00,FakcharoenpholRT04} (it is easy to go
from the tree solution to a tour by doubling each edge), they obtain
the same asymptotic approximation factor for \TSPN with line
neighborhoods. This is somewhat unsatisfactory, since it ignores that 
the group Steiner tree instances constructed by the reduction are (i)
Euclidean and (ii) all the points of a group are collinear. In other
words, although the constructed group Steiner tree instances are
highly restricted, there is no known technique to exploit 
this restriction.

However,  the reduction from \TSPN with line neighborhoods
to the group Steiner tree problem implies that, if we allow quasi-polynomial running time, then \TSPN with line neighborhoods admits an approximation ratio of $O(\log^2n/\log\log n)$ in $O(n^{\log^2 n})$-time due to the result of Chekuri and P\'al \cite{ChekuriP05}.
We would like to point out that this approximation ratio is tight for the class of quasi-polynomial time algorithms due to the recent work of Grandoni, Laekhanukit and Li \cite{GrandoniLL19}, which holds under the {\em Projection Game Conjecture} and $\mathrm{NP}\not\subseteq\bigcup_{\epsilon>0}\mathrm{BPTIME}\bigl(2^{n^{\epsilon}}\bigr)$. 
Their hardness result is built on the seminal work of Halperin and Krauthgamer \cite{HalperinK03}, %
who prove that group Steiner tree admits no $\log^{2-\epsilon}
n$-approximation for any fixed $\epsilon > 0$, unless $\NP \subseteq
\mathrm{ZTIME}\bigl(n^{\operatorname{polylog}(n)}\bigr)$.

For the class of polynomial-time approximation algorithms, the group
Steiner tree problem admits an approximation ratio of
$O(\log^2 n)$ on some special cases, e.g., trees \cite{GargKR00} and
bounded treewidth graphs \cite{ChalermsookDLV17,ChalermsookDELV18}. %
It is still open whether the group Steiner tree problem in general graphs
admits a polynomial-time $O(\log^2n)$-approximation algorithm; the best running time to obtain an $O(\log^2n)$-approximation is $n^{O(\log n)}$ \cite{ChekuriP05}.

The connection between \TSPN and group Steiner tree also holds in the reverse direction: Given an instance of group Steiner tree, one may embed the input metric into a Euclidean space with distortion $O(\log n)$~\cite{Bourgain85} and cast it as \TSPN with ``set neighborhoods''.

While we cannot improve the approximation factor in polynomial time,
we can do so in quasi-po\-ly\-nomial time: we give an $O(\log^2
n)$-approximation in $n^{O(\log\log n)}$ time. %
We obtain this result by using Arora's PTAS for \TSP~\cite{Arora98},
together with the framework of Chalermsook et
al.~\cite{ChalermsookDLV17,ChalermsookDELV18}, to transform the \TSPN problem
into a variant of group Steiner tree when the input graph is a tree, and then employing an
$O(\log^2n)$-approximation algorithm.

\subsection{Our Contribution}

Our first contribution is to show that unlike the problem with hyperplane neighborhoods, the problem with line neighborhoods is APX-hard.

\begin{theorem}\label{thm:noptas}
The \TSPN problem for lines in $\Reals^3$ is APX-hard. More specifically, it has no polynomial time $(1+\frac{1}{230000})$-approximation unless $\classP=\NP$.
\end{theorem}

The reduction is from the vertex cover problem on tripartite graphs. %
The idea is to represent the graph edges with lines, where two lines intersect
if and only if they correspond to incident edges. %
The main challenge is to keep the pairwise distance between non-intersecting
lines large enough. %
We solve this by carefully placing the intersection points on non-adjacent edges of a cube. %
For technical reasons, we do not work directly with this placement, but rather
on a ``flattened'' version of this point set. %
Additionally, we want to restrict the optimal tour so that it visits each line
near one of its intersection points with other lines. %
This is achieved by forcing the optimal tour to follow a certain closed curve using special point gadgets
(each consists of polynomially many lines), and to visit the lines representing the edges only at (or close to) intersection points. Visiting an intersection point corresponds to including the corresponding vertex in the vertex cover of the graph.
As a direct consequence of \Cref{thm:noptas}, we obtain the following.

\begin{corollary}\label{cor:kflats-apxhard}
The Euclidean TSP problem with $k$-dimensional flat neighborhoods in $\Reals^d$ is APX-hard for all $1\leq k \leq d-2$.
\end{corollary}

To prove \Cref{cor:kflats-apxhard}, suppose we are given a set $\cL$ of lines in $\Reals^3$. We can first change each line $\ell\in \cL$ into the flat $\ell\times\Reals^{k-1}$, resulting in $k$-dimensional flats in $\Reals^{k+2}$. Since $k\leq d-2$, we have that $\Reals^{k+2}$ is a subspace of $\Reals^d$, so this is a valid construction for $(k,d)$-\TSPN. Moreover, any tour in $\Reals^3$ visiting the lines is also a valid tour of the $k$-flats, and a valid tour of the $k$-flats can be projected into a valid tour of $\cL$ in $\Reals^3$ of less or equal length.

Our second contribution is to show a larger inapproximability factor in higher dimensions under the Unique Games Conjecture:

\begin{theorem}
\label{thm:tspn-hardness-line}
For any $\varepsilon>0$, there exists a constant $c$ such that there is no $(2-\varepsilon)$-approximation
algorithm for \TSPN with line neighborhoods in $\Reals^{c\cdot\log n}$, unless the Unique Games Conjecture is false. 
Moreover, for any $\varepsilon>0$, there is a constant $c$ such that it is
$\NP$-hard to give a $(\sqrt{2}-\varepsilon)$-approximation for \TSPN with
line neighborhoods in $\Reals^{c\cdot\log n}$.
\end{theorem}

This reduction is from the general vertex cover problem. Again we 
represent the edges of the graph with lines and the vertices
correspond to intersection points. This time however the intersection
points are almost equidistant: they are obtained via the
Johnson-Lindenstrauss lemma applied on an $n$-simplex. 
This allows the tour to visit the intersection points in any order. %
To obtain a direct correspondence with vertex cover, we need to ensure that
lines are visited near intersection points. %
To this end, we blow up the underlying graph by replacing each edge by a
complete bipartite graph. %
Thus, we get the following corollary of \Cref{thm:tspn-hardness-line}.

\begin{corollary}
For any $\eps>0$ there is a number $c=c(\eps)$ such that the Euclidean TSP problem with $k$-dimensional flat neighborhoods in $\Reals^d$ has no polynomial $(2-\eps)$-approximation for any $k\in\{1,\dots,d-c\log n\}$, unless the Unique Games Conjecture is false.
\end{corollary}

On the positive side, our third contribution is to develop an $O(\log^2 n)$-approximation algorithm with slightly superpolynomial running time. 

\begin{theorem}
\label{thm:tspn-algo}
There is a deterministic $O(\log^2n)$-approximation algorithm for \TSPN with line
neighborhoods in $\Reals^{d}$ that runs in time $n^{O(\log \log n)}$ for any
fixed dimension $d$.
\end{theorem}

The algorithm is based on adapting the dynamic program by Arora~\cite{Arora98}, and
reformulating \TSPN into the problem of finding a solution in the dynamic programming
space that visits all the line neighborhoods. %
We then build upon the techniques of Chalermsook et
al.~\cite{ChalermsookDLV17,ChalermsookDELV18}, and show that this task can be reduced to a
variant of the group Steiner tree problem 
that admits an 
$O(\log^2n)$-approximation in slightly superpolynomial running time.
The $O(\log \log n)$-factor in the exponent of the running time is a
consequence of the running time of Arora's algorithm, and it is
possible that we can improve it to polynomial time if an appropriate
EPTAS for \TSP with running time $O(f(\varepsilon,
d)n\log n )$ is discovered.

All missing proofs can be found in the appendix, as well as a short conclusion in Section~\ref{sec:conclusion}.

%!TEX root=./paper.tex

\section{Inapproximability in 3 dimensions.}\label{sec:noptas}

The goal of this section is to prove \Cref{thm:noptas}.
The overall setup of our construction is inspired by a reduction in
Elbassioni et al.~\cite{ElbassioniFS09} for the planar problem with
segment neighborhoods. Our reduction is from vertex cover on
$3$-partite graphs (i.e., on graphs $G$ where the vertices can be
partitioned into three independent sets $V_1,V_2$ and $V_3$). 
It is $\NP$-hard to decide whether a given instance has a vertex cover of size $n/2$ or if all vertex covers have size at least $\frac{34}{33}\frac{n}{2}$~\cite{ClementiCR99}.

In our construction, each vertex $v$ of $G$ is assigned to some point $p_v$ on some edge of a unit cube; the classes $V_1,V_2,V_3$ are mapped to pairwise non-adjacent and non-parallel (i.e., skew) cube edges. For each edge $uv\in E(G)$, we add the line $p_up_v$; see Fig.~\ref{fig:overview}.

\begin{figure}
\centering
\includegraphics[width=0.75\textwidth]{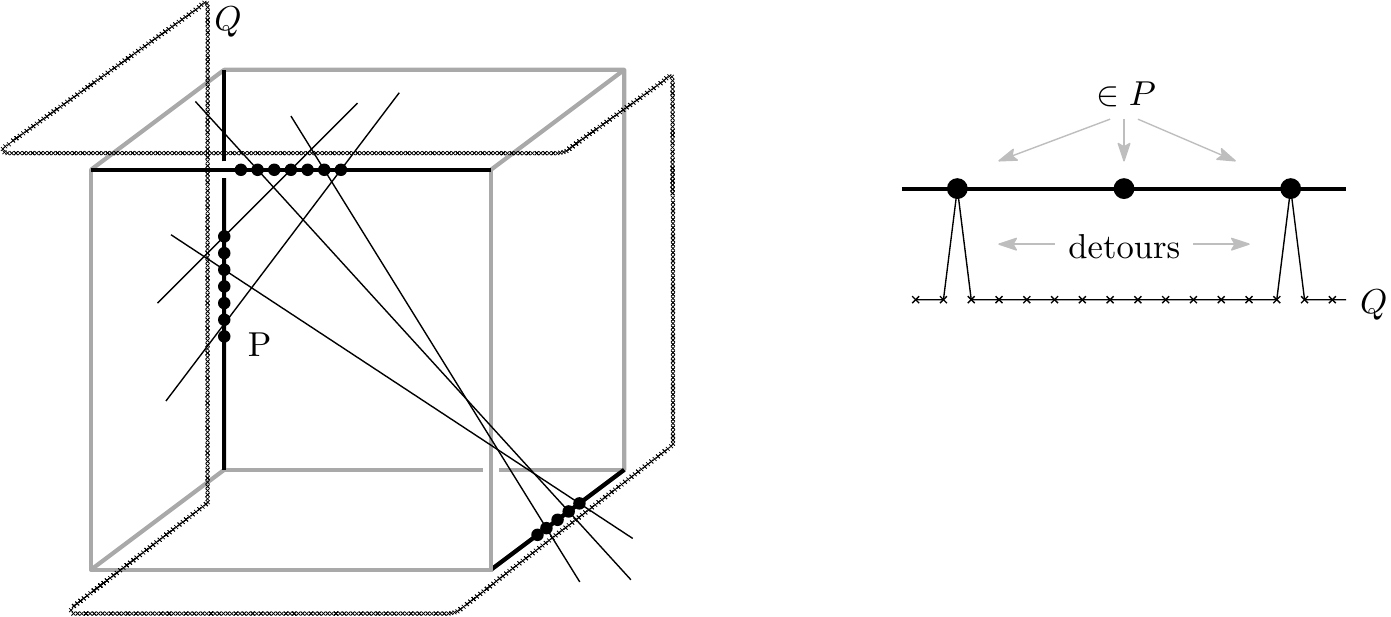}
\caption{Left: Overview of a basic construction with a cube. Right: The optimal tour must visit all points of $Q$, and it makes detours to some points $p_v$.}\label{fig:overview}
\end{figure}

Consider now a closed curve $\gamma$ of length 10 which is disjoint
form the cube, but follows some edges of the cube at a distance $c/n$
for some constant $c$. Let $Q$ be a set of points along $\gamma$ such
that any two consecutive points have distance $c/(10n)$. 

We define a special \emph{point gadget}---which consists of a large collection of lines---at each point $q\in Q$. This ensures that any TSP tour that has length at most $20$ will touch an infinitesimally small ball around each vertex of $Q$. Consequently, any not too long TSP tour will have to ``trace'' $\gamma$. The points in $P$ which are placed near the cube edges are arranged so that one can visit each point $p_v$ with a short detour from $\gamma$ of length $c/n$. Given a vertex cover of size $k$ in $G$ one can create a TSP tour of length at most $10+kc/n$, namely by folowing $\gamma$ and making the short detour at $p_v$ if and only if $v$ is in the vertex cover. Conversely, by a careful arrangement of the lines and point gadgets, we can ensure that a tour of length $10+kc/n$ implies the existence of a vertex cover of size at most $1.011 k$.

For technical reasons, we need to transform the constructed cube to a very flat parallelepiped; it is convenient to define the point set $Q$ and the point gadgets only after this flattening transformation takes place. We are now ready to define our construction.

\subsection{The construction}

Let $G=(V,E)$ be a tripartite graph on $n$ vertices with partition classes $V_1,V_2,V_3$. We add dummy vertices (without any incident edges) to G so that each class has $n$ vertices; the vertices of $V_a\, (a=1,2,3)$ are denoted by $v^a_1,\dots,v^a_n$. Notice that the addition of dummy vertices does not change the set of vertex covers of $G$. Let $\cC$ denote the unit cube $[0,1]^3$, and let $e^1,e^2,e^3$ be the unit segments $(0,0,1)^T(1,0,1)^T$, $(1,0,0)^T(1,1,0)^T$ and $(0,1,0)^T(0,1,1)^T$ respectively. We assign each vertex $v^a_i$ to a point on the middle third of $e^i$. The assignment is denoted by $p$, and defined as:
\[p(v^a_i)=
\begin{cases}
(\frac{n+i}{3n},0,1)^T & \text{if $a=1$}\\
(1,\frac{n+i}{3n},0)^T & \text{if $a=2$}\\
(0,1,\frac{n+i}{3n})^T & \text{if $a=3$}.
\end{cases}\]
We denote by $P=p(V(G))$ the set of points created this way.
For each edge $uv\in E(G)$, let $\ell(uv)$ be the line through $p(u)$ and $p(v)$, and let $\cL$ be the set of lines created this way: $\cL=\{\ell(uv) \mid uv\in E(G)\}$. The following technical lemma plays a key role in the contruction.

\begin{lemma}\label{lem:cube_skew_dist}
\begin{recallable}{lem:cube_skew_dist}
	If $\ell,\ell'\in \cL$ correspond to non-incident edges, then they are disjoint and their distance is at least~$\frac{1}{20n}$.
\end{recallable}
\end{lemma}

\subparagraph*{Flattening.}

Due to technical reasons that will become clear in the proof later on,
we need to transform the above construction so that the angle of each
line $\ell$ with the plane $x+y+z=0$ is at most some small
constant. Practically, we transform the point set $P$ and the line set
$\cL$ with the linear transformation $x \mapsto Ax$, where $A = I -
0.3J$ and $J$ is the all-ones matrix.

Essentially, the transformation pushes everything closer to the plane $H:x+y+z=0$: for a given point $p$ and its perpendicular projection $q$ on $H$, the point $Ap$ is the point on the segment $pq$ for which $\dist(q,Ap)=\frac{1}{10}\dist(q,p)$.
Note that if $pq$ is any segment of length $\lambda$, then its length after the transformation is at least $\lambda/10$ and at most $\lambda$. When the transformation is applied to an edge $e^a$ of the cube $\cC$, then the resulting segment has length $\sigma=\sqrt{0.7^2 + 0.3^2 + 0.3^2}\simeq 0.8185$. Consequently, $Ap(v^a_i)$ and $Ap(v^a_{i+1})$ has distance $\sigma/(3n)$.

Let $\bar{P}$ and $\bar{\cL}$ be the resulting point set and line set. Using \Cref{lem:cube_skew_dist} and the above arguments we get the following corollary.

\begin{corollary}\label{cor:mindist}
The minimum distance between points of $\bar{P}$ is $\frac{\sigma}{3n}$, and the minimum distance between lines of $\bar{\cL}$ corresponding to non-incident edges of $G$ is at least $\frac{1}{200n}$.
\end{corollary}

\subparagraph*{Defining the point gadgets, and wrapping up the construction.}

For a point set $X$, let $\bar{X}$ denote its image under the flattening transformation $A$.
Let $F^a_1$ and $F^a_2$ be the planes of the faces of $[0,1]^3$ incident to $e^a$. The following claim shows that $\bar{F}^a_1$ and $\bar{F}^a_2$ are two planes through $\bar{e}^a$ whose angle is small.

\begin{figure}
\centering
\includegraphics[width=\textwidth,trim={0 0 37pt 0},clip]{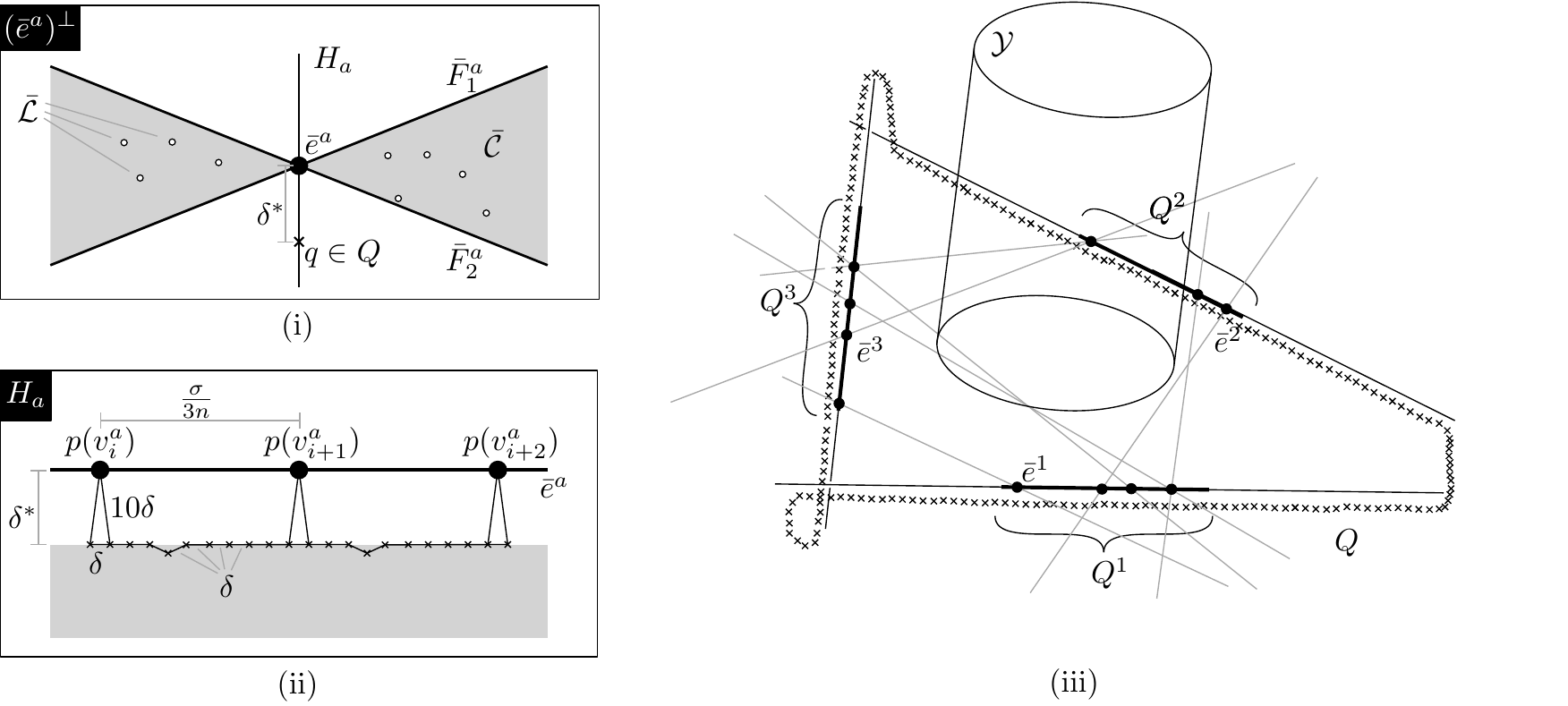}
\caption{(i) Cross-section given by a plane perpendicular to $\bar{e^a}$. (The segment $\bar{e^a}$ appears as a point, and the plane $H_a$ as a line in this picture.) All lines of $\bar{\cL}$ intersect such a plane in the gray area. (ii) Defining $Q^a$ within the plane $H_a$, so that all points have distance at least $\delta^*$ from $\bar{e^a}$. (iii) Defining $Q$ so that it has all the required properties. The cylinder $\cY$ is perpendicular to the plane $x+y+z=0$, a plane 	to which all points of the construction are close to. The ``triangle'' defined by the skew lines $\bar{e}^a$ wraps around the cylinder $\cY$.}\label{fig:Qdef}
\end{figure}

\begin{claim}\label{cl:small_plane_angle}
\begin{recallable}{cl:small_plane_angle}
For $a=1,2,3$ we have $\an(\bar{F}^a_1,\bar{F}^a_2)<\frac{1}{4}$.
\end{recallable}
\end{claim}

Let $H^a$ be the angle bisector plane of $\bar{F}^a_1$ and $\bar{F}^a_2$ which does not intersect the image of $\cC$, see Figure~\ref{fig:Qdef}(i). Within $H^a$, we place a set of points $Q^a$, which we define next.

Let $\delta=\frac{1}{4000n}$, and let $\delta^*$ be the height of the isoceles triangle $T_\delta$ with base $\delta$ and two sides of length $10\delta$, that is $\delta^*=\sqrt{99.75}\delta$. Consider a half-plane in $H_a$ whose boundary is parallel to $\bar{e}^a$ and is at distance $\delta^*$ from it. Within this half-plane, let $Q_a$ be a set of at most $4000n$ points with the following properties: (i) for each $p(v^a_i)$ there are two points $q,q'\in Q^a$ such that $p(v^a_i)$, $q$ and $q'$ form an isoceles triangle of side lengths $10\delta,10\delta,\delta$ and (ii) there is a unique shortest TSP path of $Q_a$, whose edges are of length exactly $\delta$; see Figure~\ref{fig:Qdef}(ii).

Let $Q$ be a point set with the following properties:
\begin{itemize}[noitemsep]
\item $Q^1\cup Q^2\cup Q^3 \subseteq Q$
\item For any pair of distinct points $q,q'\in Q$, $\dist(q,q')\geq \delta$.
\item Each segment of the minimum TSP tour $T(Q)$ of $Q$ has length $\delta$, and $\cost(T(Q))=10$.
\item The minimum distance of points of $Q$ from $\bar{\cL}$ is attained only in $Q^1\cup Q^2\cup Q^3$
\item $Q$ is disjoint form the cylinder $\cY$ of axis $(0,0,0)^T(1,1,1)^T$ and radius $\sigma/2$.
\end{itemize}

Such a set $Q$ is easy to find, for example by following the lines $\bar{e}^a$ and connecting them far from the origin. See Fig.~\ref{fig:Qdef}(iii) for an illustration.

We need the following claim on the distance of $Q$ from the lines in $\bar{\cL}$. %
Intuitively, it shows that the points in $Q$ are far from the lines in $\bar\cL$, and thus a certain detour is necessary to visit a line in $\bar\cL$.
Note that the bound would not be strong enough without the flattening.

\begin{claim}\label{cl:qdist}
\begin{recallable}{cl:qdist}
For any $q\in Q$ and $\ell\in \bar{\cL}$ we have $\dist(q,\ell)> 9.9\delta$.
\end{recallable}
\end{claim}

\begin{lemma}[Point gadget]
\label{lem:point_gadget}
\begin{recallable}{lem:point_gadget}
Given a positive integer $n$ and a point $q\in \Reals^3$, there is a set $\cL$ of $O(n^6)$ lines through $q$ such that any \TSPN tour of $\cL$ which is disjoint from the ball $B(q,\frac{1}{n^3})$ has length at least $20$.
\end{recallable}
\end{lemma}

Our construction is the union of the line set $\bar{\cL}$ together with a point gadget placed at each point $q\in Q$; let $\cL^*$ denote the resulting line set.

\subsection{The Reduction}

\begin{lemma}\label{lem:noptas_main}
\begin{recallable}{lem:noptas_main}
If $G$ has a vertex cover of size $k$, then there is a tour in $\cL^*$ of length $10+19\delta k$. If $\cL^*$ has a tour of length $10+19\delta k$, then $G$ has a vertex cover of size $1.011k$.
\end{recallable}
\end{lemma}

The proof of the first part of the lemma is straightforward. To prove the second claim, we use the fact that the tour must touch the small balls $B(q,\frac{1}{n^3})$ for each point $q\in Q$ by \Cref{lem:point_gadget}. We can then consider a portion of the tour between two consecutive ball visits, i.e., a polygonal curve $g$ that starts near some point $q\in Q$ and ends near some other point $q'\in Q$, and visits some of the lines in $\bar{\cL}$ along the way. In \Cref{lem:cylinder} we show that $g$ cannot touch lines from all three classes, in other words there is a segment $\bar{e}^i$ such that all lines visited by $g$ have an endpoint on $\bar{e}^i$. The proof of \Cref{lem:cylinder} relies on the property that $Q$ avoids the cylinder $\cY$ with axis $(0,0,0)^T,(1,1,1)^T$ and radius $\sigma/2$. Intuitively, if $g$ would touch lines from all three classes, then it would have to go around the cylinder partially, which would be too costly. We can then define a vertex cover based on the tour portions $g$: for each line $\ell$ visited by $g$, the line $\ell$ has a point on $\bar{e}^i$ that corresponds to some vertex $v$ of the graph. These vertices $v$ form a set $W$ which is clearly a vertex cover; the goal is then to prove that $|W|\leq 1.011k$. The proof hinges on the fact that if a tour portion $g$ contributes $s$ unique vertices to $W$, then it must jump between non-incident lines of $\bar{\cL}$ at least $(s-1)$ times, which incurs a cost of at least $20(s-1)\delta$ by \Cref{cor:mindist}. In case of $s=1$, the tour still needs to visit \emph{some} line in $\bar{\cL}$, which incurs a cost of at least $19.8\delta$ by \Cref{cl:qdist}. Putting these observations together (and that the minimum cost tour of the balls $B(q,1/n^3)$ has length very close to $10$) yields the desired bound on $|W|$.

\begin{proof}[Proof of \Cref{thm:noptas}] Suppose that there is a polynomial time algorithm that approximates \TSPN with lines in $\Reals^3$ within a factor of $1+\frac{1}{230000}$. Let $G$ be a given $3$-partite graph. If $G$ has a vertex cover of size $n/2$, then the above construction would have a tour of length $10+19\delta \frac{n}{2}%
% = 10+\frac{19}{8000} 
= 10.002375$. On the other hand, if all vertex covers of $G$ have size at least $\frac{34}{33}\frac{n}{2}$, then all tours of the construction have length at least $10+19\delta\frac{34}{33}\frac{n}{2\cdot1.011} >%
%= 10 + \frac{34\cdot 19}{33\cdot1.011\cdot8000}>
10.00242$. As $10.00242/10.002375>1+\frac{1}{230000}$, we could use the hypothetical approximation algorithm to distinguish between these two cases in polynomial time, which would imply $\classP=\NP$.
\end{proof}

%!TEX root = ./paper.tex

\section{No \texorpdfstring{$(2-\epsilon)$}{(2-e)}-approximation Algorithm}
\label{sec:hardness-line}

In this section we prove \Cref{thm:tspn-hardness-line}.
In particular, we will show that when the objects are lines, \TSPN is at least as hard to approximate as the Vertex Cover problem
which is known to be hard to approximate to within a factor of
$2-\varepsilon$, for any constant $\varepsilon>0$, under the Unique
Games Conjecture (and inapproximable within a factor of $1.42$ unless $\classP=\NP$~\cite{KhotMS18}).

\begin{theorem}[\cite{KhotR08}]
\label{thm:ugc-vertex-cover}
Unless the Unique Games Conjecture is false,
for any constant $\varepsilon>0$,
there is no polynomial-time algorithm that, 
given a graph $G=(V,E)$ and an integer $k$, 
distinguishes between the cases (i) $G$ has a vertex cover of size at most $k$ or (ii)
$G$ has no vertex cover of size less than $(2-\varepsilon)k$.
\end{theorem}

\noindent The main idea behind the reduction is to represent a graph $G$ in Euclidean space such that:
\begin{itemize}[noitemsep]
  \item Each vertex $v\in V(G)$ corresponds to a point $p_v\in \Reals^d$,
  \item Each edge $e = uv\in E(G)$ corresponds to a line going through the
    points $p_u$ and $p_v$,
  \item An optimal tour visits each line sufficiently
    close to the points
    $p_v$, and therefore the vertex set corresponding to the points in
    the vicinity of the tour is a vertex cover.
\end{itemize}

However, in order to enforce that an optimal tour passes through (or not too far
from) the points $p_v$, we will have to further build upon this
 idea. In particular, for each vertex $v$, instead of
constructing only one point $p_v$, we will construct a set $P_v$ of
polynomially many points corresponding to $v$. %
If there is an edge $e=uw\in E(G)$, then we connect each point corresponding
to $u$ with each point corresponding to $w$. More precisely, for each edge
$uv$ and for every pair of points $(p_u, p_w)$ with $p_u\in P_u$ and $p_w\in
P_w$, we add a line going through $p_u$ and $p_w$. %
Notice that the number of edges increases quadratically in the number of vertex copies. %
Therefore, tours that visit lines away from the vertices are disproportionally affected,
which  forces an optimal
tour to visit lines at (or close to) the points in $P_v$. 

Another key aspect of our
construction is that we position the points of $\mathcal{P} := \bigcup_{v\in V(G)}P_v$
in $\Reals^d$ so that the distance between any pair of
distinct points is (roughly) the same. This helps us to have a more direct
correspondence between the cost of the optimal tour and the size of
an optimal vertex cover. The reduction is desribed formally in
the next subsection.

\subsection{Reduction: Vertex Cover to TSP with Line Neighborhoods}
\label{sec:hardness-line:reduction}

Take an instance of the Vertex Cover problem
on a graph $G=(V,E)$ with $n$ vertices and $m$ edges.
We first take a {\em lexicographic product} of the graph $G$ with an independent set of size $\alpha=n^2$.
Informally speaking, we construct a graph $G'$ by making $\alpha$
copies of each vertex $v\in V(G)$, and denote the corresponding vertex
set by $Q_v$.
Then, for each edge $vw\in E(G)$, we add edges between every pair of vertices $v^i\in Q_v$ and $w^j\in Q_w$, thus forming a complete bipartite graph on $Q_v$ and $Q_w$.
More formally, the graph $G'$ is defined as:
\begin{align*}
V(G') = \{v^i: v\in V(G) \land i\in [\alpha] \} \text{\quad and \quad} E(G') = \{v^iw^j: vw\in E(G) \land v\neq w \land i,j\in[\alpha] \}.
\end{align*}

Next, we use the graph $G'$ to construct an instance $I(G')$ of the
\TSPN with line neighborhoods problem in $d=O(\delta^{-2}\ln n')$
dimensions for any small enough $\delta>0$ and with $n'=|V(G')| =\alpha\cdot n$.  We
map each vertex $v$ of $G'$ to a point $p_v$ in $\Reals^d$ such that for any two points $p_v$ and $p_u$ with $v\neq
u, v,u\in V(G')$ the distance $\dist(p_v,p_u)$ between them satisfies
the following property: $1 \le \dist(p_v,p_u) \le 1+\delta$.

The fact that this is possible and can be done in polynomial time follows
by Theorem~3.1 by Engebretsen, Indyk and
O'Donell~\cite{EngebretsenIO02}. 
In particular, we
can employ the theorem in order to deterministically map a unit side length simplex from $\Reals^{n'-1}$ to $\Reals^d$
such that the desired property holds for all pairs of points.

 We denote the resulting point set by $P$.
Next, we create a collection of lines $\mathcal{L}$ in an instance of \TSPN,
by adding to $\mathcal{L}$ a line $\ell_{vw}$ passing through 
points $p_v$ and $p_w$ if $vw\in E(G')$.

We devote the rest of this section to prove completeness and
soundness of our reduction.

\subparagraph*{Completeness.}
Suppose the graph $G$ has a vertex cover of size $\leq k$.
Then we claim that there is a tour $T$ of cost at most
$\alpha k (1+\delta)$ that touches each line at least once.
To see this, let $S = \{v_1,\ldots,v_k\}$ denote the vertex cover of $G$.
By construction, $S'=\{v^i_{j}:i\in[\alpha] \land j \in [k]\}$ is
a vertex cover of $G'$. By the construction of $\cL$ and by the fact that $S'$
is a vertex cover of $G'$, it follows that any tour that visits 
points $p_{v_1^1},p_{v_1^2},\dots p_{v_k^\alpha}$ (in any order) is a
feasible tour, i.e., it touches all lines in $\mathcal{L}$.
So, in total such a tour visits a total of at most $\alpha k$ points, and
the distance between any pair of these points is by construction at most
$1+\delta$. 
Thus, there is a solution to \TSPN with cost at most $\alpha k (1+\delta)$.

\subparagraph*{Soundness.}

We show that if there is a tour of cost $x$ (where $x\leq \alpha n (1+\delta)$), then there is a vertex cover in $G$ of size at most $\frac{x}{\alpha (1-2\Delta) \lambda}$, where $\Delta$ is a small positive number and $\lambda\in [0,1]$ is very close to $1$.

The intuition behind
$\Delta$ is that it describes the maximum distance that the tour is
allowed to have to a
given point, assuming that the vertex corresponding to that point
contributes 
to the vertex cover.
 For each point $p_{v^i}\in P$ (note
that $v^i\in V(G')$), let $B(v^i)$ be a $d$-dimensional ball of radius $\Delta$
centered at $p_{v^i}$.  %
Note that $\Delta$ is small enough so the only lines from $\cL$ intersecting a ball
$B(v^i)$ are the ones that go through $p_{v^i}$. Given a tour $T$, we say that a ball
$B(v^i)$ is {\em non-empty} if $T\cap B(v^i) \neq \emptyset$;
otherwise, we say that $B(v^i)$ is {\em empty}.  
We say that a line $\ell_{uw}$ is \emph{covered by a ball} if at least
one of the balls $B(u)$ and $B(w)$ is non-empty. Otherwise
$\ell_{uw}$ is \emph{not covered by a ball}. 
We first show that any point $p\in\ell_{uw}$ that is outside the
two balls corresponding to $u$ and $w$ will not be ``too close'' to
any other line:
\begin{lemma}
  \label{lem:distance_point_line}
  \begin{recallable}{lem:distance_point_line}
    For any point $p\in\ell_{uw}$ such that $p\not\in B(u)$ and $p\not\in
    B(w)$ and for any $\ell\in \cL\setminus \{\ell_{uw}\}$ we have $dist(p,\ell)\ge \Delta/2$.
  \end{recallable}
\end{lemma}

We are now ready to prove that any optimal tour $T$ must cover almost all lines by
balls:
\begin{lemma}
  \label{lem:covering_at_balls}
  \begin{recallable}{lem:covering_at_balls}
  Let $T$ be a tour of cost at most $x$ with  $ x \le \alpha n
  (1+\delta)$ for the instance $I(G')$. 
Then the number of lines of $I(G')$ that are
   not covered by balls is at most
   $\frac{2x}{\Delta}$.
  \end{recallable}
\end{lemma}

Let $\lambda = 1 - \eps^2$ and set $\alpha=n^2$. We can construct a vertex cover of $G$ based on a tour $T$ the following way: if a set $Q_v$ has at least $\lambda\alpha$ non-empty balls, then we add $v$ to the vertex cover.

\begin{lemma}
\label{lemma:non-empty-balls}
\begin{recallable}{lemma:non-empty-balls}
The set $S=\{v: \lvert \cup_{i\in[a]}\{ v^i: B(v^i) \text{ is non-empty}\}\rvert
\ge \lambda\alpha\}$ is a vertex cover
of $G$ of size $|S| < \frac{x}{\alpha (1-2\Delta) \lambda} $.
\end{recallable}
\end{lemma}

\begin{proof}[Proof of \Cref{thm:tspn-hardness-line}]
Suppose that there is an algorithm that can distinguish in polynomial time, for any $0<\eps'$ and any $x\in \Reals_+$, whether there is a tour of length at most $x$ or all tours have length at least $(2-\eps')x$. Take some instance of vertex cover, where the goal is to decide if there is a vertex cover of size at most $k$ or all vertex covers of the graph have size at least $(2-\eps)k$, where $\eps\in (0,0.1]$. By the above polynomial construction, it would be sufficient to distinguish the cases where $\mathcal{I}(G')$ has a tour of size at most $k \alpha (1-2\Delta) \lambda$ (implying a vertex cover of size at most $k$), or all tours have length at least $(2-\eps)k\alpha(1+\delta)$ (implying that all vertex covers have size at least $(2-\eps)k$). If we set $\delta=\Delta=\eps^2$, then we get that the ratio of these tours is:
\[\frac{(2-\eps)k\alpha (1+\delta)}{k\alpha (1-2\Delta)\lambda}=\frac{(2-\eps)(1+\eps^2)}{(1-2\eps^2)(1-\eps^2)}<2,\]
so the hypothetical algorithm on $I(G')$ distinguishes
these cases, which is a contradiction.
\end{proof}

We note that our reduction implies that \TSPN with Line Neighborhoods
is Vertex Cover hard, and therefore also inapproximable within a
factor of $\sqrt{2}-\eps$ unless $P=NP$~\cite{KhotMS18}.

%%% Local Variables: 
%%% mode: latex
%%% TeX-master: "paper.tex"
%%% End: 

%!TEX root=./paper.tex

\section{A Superpolynomial-Time Approximation Algorithm}
\label{sec:super}

\newcommand{\strue}{\textsf{True}\xspace}
\newcommand{\sfalse}{\textsf{False}\xspace}

In this section, we will show a quasi-polynomial time algorithm to approximate
\TSPN for lines to a factor of $O(\log^2n)$. %
In fact, our approach is more general: we show how to $O(\log{}N\log{}n)$-approximate \TSPN for discrete neighborhoods of total size $N$, %
in running time $N^{O(\log \log N)}$ for any fixed $d$. 
In this problem, we are given $n$ neighborhoods $P_i \subset \Reals^d$, which are discrete
sets of points. We denote by $P = \bigcup_{i \in [n]} P_i$ the union of all neighborhoods, and by $N = |P|$ its size. %
Using the approach of Dumitrescu and Tóth~\cite{DumitrescuT16}, we can convert
any instance of \TSPN with line neighborhoods into an instance of discrete $\TSPN$ on a
set of $N=O(n^4)$ points and $n$ neighborhoods. %
This transformation has a running time of $O(N)$, and incurs the loss of a constant factor in the approximation.
From now on, we focus on $\TSPN$ for discrete neighborhoods.

Our main result is an $O(\log N \log n)$-approximation algorithm for that runs
in time $N^{O(\log \log N)}$ for constant $d$. %
Our algorithm combines the dynamic program by Arora~\cite{Arora98} with the
framework of Chalermsook et al.~\cite{ChalermsookDLV17,ChalermsookDELV18}. %
As Dumitrescu and Tóth show \cite{DumitrescuT16}, \TSPN is related to the
group Steiner tree problem, and can be reduced to this problem to obtain an
$O(\log^3 n)$-approximation. %
We show that, using the structure of the Euclidean space, which is exploited
in the algorithm presented by Arora for \TSP, we can use the techniques of Chalermsook et al.~to approximate discrete
instances of \TSPN and group Steiner tree in $\Reals^d$. %
Notice that, even on tree metrics, the group Steiner tree
problem is $\Omega(\log^2n/\log\log n)$-hard to approximate~\cite{HalperinK03,GrandoniLL19} under the projection games conjecture. %
As every tree metric can be embedded into some Euclidean space with distortion $O(\sqrt{\log \log n})$ \cite{LinialMS98,Matouvsek99,Matousek2013Book}, the group Steiner tree problem in Euclidean space is also hard to approximate to within $\Omega(\log^2n/(\log\log n)^{3/2})$ under the same assumption. %

\begin{theorem}
\label{thm:super:thm}
\begin{recallable}{thm:super:thm}
There is a randomized $O(\log N \log n)$-approximation algorithm for \TSPN with discrete
neighborhoods in $\Reals^{d}$ that runs in time $N^{O(\log \log N)}$ for
constant $d$.
\end{recallable}
\end{theorem}

The theorem above, together with the result of Dumitrescu and
T\'oth~\cite{DumitrescuT16} imply \Cref{thm:tspn-algo},
along with the derandomization techniques of Arora~\cite{Arora98} and Charikar et al.~\cite{CharikarCGG98}.

We start by recalling the main steps of the PTAS for \TSP by Arora, as
our result builds upon the dynamic program used there. %
While describing the algorithm, we state some modifications
that are necessary for our purpose. %
Then, we show how to use the framework of Chalermsook et
al.~\cite{ChalermsookDLV17,ChalermsookDELV18} to find a feasible solution using the
dynamic program.

Among all the PTASes for Euclidean \TSP, we choose to base our algorithm
on the work of Arora, as it results in the lowest running time for our
algorithm. %
Unfortunately, the results of Rao and Smith~\cite{RaoSmith98}, and Bartal and
Gottlieb~\cite{BartalG13} cannot be adapted for our purposes, %
since their algorithms use spanners to reduce the total weight of the graph to
be a constant factor away from the optimum. %
This technique breaks down for discrete \TSPN, as the spanner contains the
entire set $P$ of points, and a minimum-cost tree spanning $P$ may be much
larger than the optimum solution. %

\subsection{Arora's Algorithm}
\label{sec:super:arora}

In this section, we shortly summarize the algorithm of
Arora~\cite{Arora98} 
(a more detailed description can be found in \Cref{sec:app:super:arora}). %
This algorithm approximates \TSP to a factor of $1+1/c$; for our purpose, it is sufficient to consider $c=1$. %
Arora's algorithm has three main steps:
\begin{enumerate}[noitemsep]
\item Perturbation, which makes all coordinates are integral and bounded by $O(n)$;
\item Construction of a shifted quadtree;
\item Dynamic program, which finds the approximate solution for TSP.
\end{enumerate}

The dynamic program is based on the \emph{$(m,r)$-multipath problem} (see
\Cref{superpoly:multipath}), which given a cell of the quadtree and a set of
pairs of \emph{portals} on the boundary of the cell, has as its objective to
find a minimum-cost set of paths, each connecting a pair of portals, and such
that all of the points in the cell are visited. %
We refer to the multiset of
portals and their pairing as the \emph{state} of an $(m,r)$-multipath problem.

Two main changes are required to use the algorithm by Arora to approximate \TSPN. %
First, we must guess a point $v_0$ in an optimum solution, as well as a value
$R = O(\OPT)$ (see \Cref{sec:super:arora:1}). %
Second, we must allow the solutions to the $(m,r)$-multipath problem to not visit every point in the cell. %
We achieve this by adding a \emph{visit bit} to the state of the $(m,r)$-multipath problem on leaves: if the visit bit is $\strue$, the solution must visit the (unique) point in the cell; otherwise, it is only forced to connect each pair of portals (see \Cref{sec:super:arora:3}. %
Using these modifications, we can prove the result of \Cref{thm:super:thm}.

\subsection{Approximating \TSPN using the framework by Chalermsook et al.}
\label{sec:super:approx}

After perturbation and construction of the shifted quadtree, we use the
dynamic program above to define a \emph{dynamic programming graph}. %
The intuition is that a solution to the problem can be represented as a tree
in this graph, where the vertices in the tree correspond to all of the
$(m,r)$-multipath problems that assemble into the solution. %

We now describe the nodes and edges of this graph, denoted by $H$. %
\begin{itemize}[noitemsep]
\item \textbf{Nodes:} There are two types of nodes, which we refer to as \emph{subproblem nodes} and \emph{combination nodes}. %
The graph contains one subproblem node for every entry of the modified dynamic
programming table in \Cref{sec:super:arora}, that is, one node for each instance of the
$(m,r)$-multipath problem for every cell and set of portals and their
pairings. %
Combination nodes correspond to the possibilities of recursion for a given
subproblem: %
for a given $(m,r)$-multipath problem (for a non-leaf cell),
there is a combination node for every possible way for the $p$
paths to cross the boundary between children cells. %
\item \textbf{Root:} The root of $H$ corresponds to $(m,r)$-multipath on the root cell with no portals. %
\item \textbf{Edges:} There are (directed) edges connecting the node for each
$(m,r)$-multipath problem to the corresponding combination nodes, and then the
combination nodes to the corresponding nodes for the subproblems in the children cells. %
\item \textbf{Costs:} Edges incident to leaf nodes have cost equal to the corresponding entry in the dynamic programming table; all other edges have cost 0.
\end{itemize}

Using this definition, we can represent any $(m,r)$-light salesman path as a
tree $T$ in $H$. %
For each cell, the solution restricted to that cell consists of a union of
disjoint paths, which induce a set of portals and their pairing, and hence an
instance of the $(m,r)$-multipath problem. %
We include the corresponding subproblem node in $T$. %
For each non-leaf cell, there is a
 combination node which represents the way in which the paths
cross boundaries between children cells. %
We add that combination node to $T$, as well as all of the edges containing
it.

The trees obtained by this process have a very specific structure: the root
node is always included, as well as exactly one subproblem node for each cell
and one combination node for each non-leaf cell. %
Adding the edges between these nodes, we see that each node other than the
root has in-degree $1$, each subproblem node has exactly one outgoing edge (if
it is not a leaf), and each combination node has full out-degree, as all of its
children nodes are also in the solution.
All of these properties are implicitly formulated in the work of Chalermsook
et al.~\cite{ChalermsookDLV17,ChalermsookDELV18}; we formalize them below.

\newcommand{\Scal}[0]{\ensuremath{\mathcal S}\xspace}
\begin{definition}[Solution tree] \leavevmode
\label{def:gst:btw:soltree}

Let $H$ be a DAG with root $r$, and its nodes be partitioned into
\emph{combination nodes} $H_c$, and \emph{subproblem nodes} $H_p$.
We say an out-arborescence $T\subseteq H$ rooted at $r$ is a \emph{solution tree} if:
\begin{enumerate}[noitemsep]
\item Every combination node $t^c \in T \cap H_c$ has full out-degree (i.e.,~all children are also in~$T$),
\item Every non-leaf subproblem node $t \in T \cap H_p$ (including the root $r$) has
out-degree $1$ in $T$. 
\end{enumerate}
\end{definition}

As we mentioned above, we can associate a solution tree to any $(m,r)$-light
salesman path. The converse is also true: for each solution tree, there is a
corresponding $(m,r)$-light salesman path. %
The final requirement for a solution to be feasible is that each neighborhood
must be covered, meaning that the tour must intersect every neighborhood. %

Consider a tour corresponding to a solution tree in $T$. %
If a leaf subproblem node in the solution tree corresponds to an
$(m,r)$-multipath problem with the visit bit set to $\strue$, then the (unique) point in the cell is visited. %
In fact, the set of points visited by this tour is exactly the set of points contained in the leaf cells for subproblem nodes with the corresponding visit bit set to \strue. %
In other words, including a given leaf subproblem node in the solution tree
implies that certain neighborhoods are covered by the solution.

We can solve \TSPN by formulating it as finding a solution tree that covers every neighborhood. %
Let $S_i$ be the set of all subproblem nodes whose visit bit is \strue, and whose cell contains a point in $P_i$. %
Our goal is to find the minimum-cost solution tree that contains at least one node of each set $S_i$. %
This problem resembles \GST, and is defined in the work of Chalermsook et al.~\cite{ChalermsookDLV17,ChalermsookDELV18}. %
We redefine the problem using our own notation.

\begin{definition}[Solution Tree Group Steiner Tree (\STGST)] \leavevmode
\label{def:super:approx:stgst}

Let $H$ be a DAG with edge-costs $\cost: E(H) \to \Reals$ and root $r$, as well as \emph{groups} $S_i \subseteq V(H)$, for $i \in [h]$, and a partition of the nodes into \emph{combination} and \emph{subproblem} nodes ($H_c$ and $H_p$ respectively).
The objective of this problem is to find a minimum-cost solution tree $T$ that
contains at least one vertex of every group $S_i$.
\end{definition}

Their work shows that we can approximate this problem on DAGs, in the following sense.

\begin{theorem}[{\cite{ChalermsookDLV17,ChalermsookDELV18}}]
\label{thm:super:approx:round}
Let $H$ be a DAG with edge-costs $\cost: E(H) \to \Reals$ and root $r$, as well as {groups} $S_i \subseteq V(H)$, for $i \in [h]$, and a partition of the nodes into $H_c$ and $H_p$.
There is an algorithm that outputs a solution tree $X \subseteq H$ sampled
from a distribution $\mathcal{D}$ such that:
\begin{enumerate}[noitemsep]
\item $\E_{X\sim \mathcal{D}}[\cost(X)] \leq \cost(\OPT)$, where $\cost(\OPT)$ denotes the cost of the optimal solution
\item For any group $S_i$, the probability that the group is covered (for some constant $\alpha > 1$) is
\[
\Prob_{X \sim \mathcal{D}}\bigl[|S_i \cap X| > 0\bigr] \geq \frac{1}{\alpha \height(H)}
\]
\end{enumerate}
The running time of this algorithm is $\Delta(H)^{O(\height(H))}$, where $\Delta(H)$, $\height(H)$ are the maximum out-degree and height of $H$, respectively.
\end{theorem}

Using this result, all we need to prove \Cref{thm:super:thm} is to
show that we can formulate \TSPN as an instance of \STGST, and then to show how
to obtain an $O(\log^2 n)$-approximation from \Cref{thm:super:approx:round}. %
We show the details of these steps in \Cref{sec:app:super:approx}.

%!TEX root=./paper.tex

\section{Conclusion}\label{sec:conclusion}

We have shown that \TSPN with line neighborhoods is APX-hard, so a PTAS for this problem is unlikely. This implies the same hardness for $k$-dimensional
flats in $\Reals^d$ for $1\leq k \leq d-2$, which together with the known PTAS results for $k=0$ and $k=d-1$ gives a complete classification of these
problems. We have also proved a stronger inapproximability factor for
$d=O(\log n)$: there is no $(\sqrt{2}-\eps)$-approximation assuming $\classP \neq
\NP$ and no $(2-\epsilon)$-approximation assuming the UGC. On the positive side, we gave
an $O(\log^2 n)$-approximation algorithm in slightly superpolynomial
time.

There is still a large gap between the lower bounds and the algorithms
for \TSPN with line neighborhoods. Perhaps the most important question
related to \TSPN is to find a constant-approximation for line neighborhoods in $\Reals^3$, or to prove that it does not exist.
Furthermore, for general point sets in higher dimensions there is an
inapproximability of $\Omega(\log^2n/(\log\log n)^{3/2})$ under the Projection
Games Conjecture. Whether that holds for flats or lines is an open problem.

\nolinenumbers
\bibliographystyle{plain}
\bibliography{tspn}

\appendix
%!TEX root=./paper.tex

%%% Local Variables: 
%%% mode: latex
%%% TeX-master: t
%%% End: 

\section{Detailed Proofs from Section~\ref{sec:noptas}}

\begin{figure}
\centering
\includegraphics{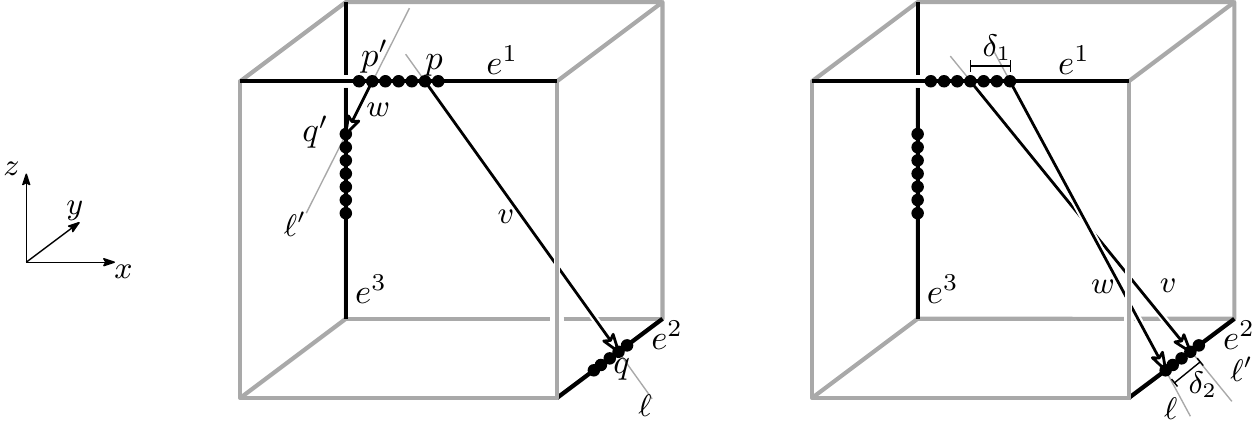}
\caption{Bounding the distance of the skew lines $\ell,\ell'\in \cL$.}\label{fig:linedist}
\end{figure}

\recallthm{lem:cube_skew_dist}

\begin{proof}%[Proof of Lemma~\ref{lem:cube_skew_dist}]
Assume without loss of generality that  $\ell$ and $\ell'$ intersect $e^1$ at the points $p,p'$ respectively, and moreover assume that $\ell$ also intersects $e^2$ at $q$. The line $\ell'$ intersects either $e^2$ or $e^3$ at some point $q'$, see Figure~\ref{fig:linedist}.
For a vector $x\in \Reals^3$, let $x_1,x_2$ and $x_3$ be its coordinates. The line $\ell$ has the vector equation $\ell=\{p+\xi (q-p) \mid \xi \in \Reals\}$, and similarly $\ell'=\{p'+\xi (q'-p') \mid \xi \in \Reals\}$. Their distance is therefore
\[\dist(\ell,\ell')=\Big|\Big\langle p'-p, \, \frac{(q-p)\times(q'-p')}{|(q-p)\times(q'-p')|} \Big\rangle\Big|.\]
Let $v=q-p$ and let $w=q'-p'$. Since $p'-p$ is parallel to $e^1$ and the $x$-axis, this is
\[\dist(\ell,\ell')=|p'_1-p_1| \frac{|(v\times w)_1|}{|v\times w|}=\frac{|(p'_1-p_1)(v_2w_3-v_3w_2)|}{\sqrt{(v_2w_3-v_3w_2)^2+(v_3w_1-v_1w_3)^2+(v_1w_2-v_2w_1)^2}}.\]
Note that $p_1'-p_2'\geq \frac{1}{3n}$, the values $v_1$ and $v_2$ are both in the interval $[1/3,2/3]$, and $v_3=-1$.
We now consider the cases where $\ell'$ intersects $e^2$ and $e^3$ separately.
If $\ell'$ intersects $e^3$, then $w_2=1$, and $w_1$ and $w_3$ are inside the interval $[-2/3,-1/3]$. We can use these facts to bound each term:
\[\dist(\ell,\ell')\geq \frac{\frac{1}{3n}(-2/3+1)}{\sqrt{(-1/9+1)^2+(2/3+4/9)^2+(2/3+4/9)^2}}=\frac{1}{2\sqrt{66}n}>\frac{1}{20n}.\]
If $\ell'$ intersects $e^2$, then $w_1$ and $w_2$ are both in the interval $[1/3,2/3]$, and $w_3=-1$. As we now have $q_1=q'_1=1$, it follows that $p'_1-p_1= (p'_1 - q'_1) -(p_1-q_1) = v_1 - w_1$. Let $\delta_i=v_i-w_i$. Then the above formula becomes:
\[\dist(\ell,\ell')=\frac{|\delta_1\delta_2|}{\sqrt{\delta_1^2+\delta_2^2+(v_1w_2-v_2w_1)^2}}.\]
The term $|v_1w_2-v_2w_1|$ can be bounded the following way:
\[|v_1w_2-v_2w_1|= \frac{|(v_1-w_1)(v_2+w_2)-(v_1+w_1)(v_2-w_2)|}{2}= \frac{|\delta_1(v_2+w_2)-\delta_2(v_1+w_1)|}{2}  \leq |\delta_1|+|\delta_2|.\]
We can substitute this to get a lower bound on the distance:
\begin{align*}\dist(\ell,\ell')
&\geq \frac{|\delta_1\delta_2|}{\sqrt{\delta_1^2+\delta_2^2+(|\delta_1|+|\delta_2|)^2}}
>  \frac{|\delta_1\delta_2|}{\sqrt{2(|\delta_1|+|\delta_2|)^2}} 
= \frac{1}{\sqrt{2}}\,\frac{|\delta_1\delta_2|}{|\delta_1|+|\delta_2|}\\
&= \frac{1}{2\sqrt{2}}\frac{2}{\frac{1}{|\delta_1|}+\frac{1}{|\delta_2|}}
\geq \frac{1}{2\sqrt{2}}\min(|\delta_1|,|\delta_2|)
\geq \frac{1}{6\sqrt{2}n}>\frac{1}{10n},
\end{align*}
where we have used that both $|\delta_1|$ and $|\delta_2|$ is at least $\frac{1}{3n}$, and the fact that the harmonic mean of two numbers is at least as large as the smaller number.
\end{proof}

\recallthm{cl:small_plane_angle}

\begin{proof}%[Proof of \Cref{cl:small_plane_angle}] 
It is sufficient to prove the claim for $a=1$ because of symmetry. First, we compute the normal of $\bar{F}^1_1$, which is the plane through the points $A(0,0,1)^T$, $A(1,0,1)^T$, and $A(0,0,0)^T$. Therefore, it goes through $p_1=(-0.3,-0.3,0.7)^T$, $p_2=(0.4,-0.6,0.4)^T$ and the origin; its normal is therefore $n_1=\frac{p_1 \times p_2}{|p_1 \times p_2|}=\frac{1}{\sqrt{0.34}}(0.3,0.4,0.3)^T$. Similarly, $\bar{F}^1_2$ is the plane through $p_1$, $p_2$ and $A(1,1,1)^T=(0.1,0.1,0.1)$. The calculation yields that the normal is $n_2=\frac{1}{\sqrt{0.34}}(0.3,0.3,0.4)^T$. The angle of the planes is therefore
\[\cos^{-1}(\langle n_1,n_2\rangle)=\cos^{-1}(0.33/0.34)<\frac{1}{4}.\qedhere\]
\end{proof}

\recallthm{cl:qdist}

\begin{proof}%[Proof of \Cref{cl:qdist}]
As the distance between points of $Q$ and lines in $\bar{\cL}$ is minimized only at points of $Q_a$, we may assume without loss of generality that $q\in Q_1$. If $\ell$ connects $\bar{e}^2$ and $\bar{e}^3$, then its distance from $q$ is much more than $9\delta$, so assume that $\ell$ connects a point of $\bar{e}^1$ to $\bar{e}^2$. (The case when $\ell$ goes from $\bar{e}^1$ to $\bar{e}^3$ is similar.) Notice that all such lines are separated from $q$ by the union of the planes $\bar{F}^a_1$ and $\bar{F}^a_2$, so $d(q,\ell)\geq \min(d(q,\bar{F}^1_1),d(q,\bar{F}^1_2))$. By the definition of $Q_1$, we have that
\[d(q,\bar{F}^1_1)=d(q,\bar{F}^1_2)=\dist(q,\bar{e}^1)\cos(\an(\bar{F}^a_1,\bar{F}^a_2)/2).\]
Since $\an(\bar{F}^a_1,\bar{F}^a_2)< 1/4$ by \Cref{cl:small_plane_angle}, and $\dist(q,\bar{e}^1)=\sqrt{(10\delta^2)-(\delta/2)^2}>9.98\delta$, therefore
\[d(q,\ell) > \cos(1/8) \cdot 9.98\delta > 9.9 \delta\qedhere\]
\end{proof}

\begin{figure}
\centering
\includegraphics[scale=0.9]{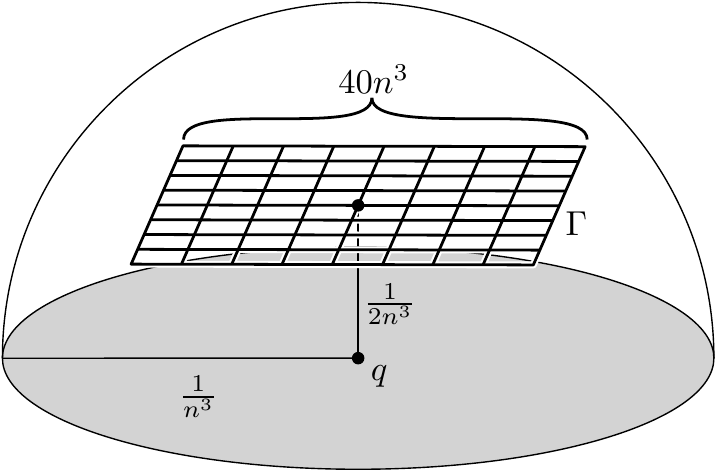}
\caption{Defining the lines of a point gadget for a point $q$ using a grid $\Gamma$.}\label{fig:pointgagdget}
\end{figure}

\recallthm{lem:point_gadget}

\begin{proof}%[Proof of \Cref{lem:point_gadget}]
Let $\Gamma$ be a $40n^3\times 40n^3$ plane grid where each cell has side length $\frac{1}{80n^6}$. The grid fits in a square of side length $\frac{1}{2n^3}$. Let $q$ be the origin, and place the grid $\Gamma$ in the plane $z=\frac{1}{2n^3}$, within the axis-parallel square with diagonal vertices $(-\frac{1}{4n^3},-\frac{1}{4n^3},\frac{1}{2n^3})^T$ and $(\frac{1}{4n^3},\frac{1}{4n^3},\frac{1}{2n^3})^T$. See Figure~\ref{fig:pointgagdget}. Notice that the grid is contained in the ball $B(q,\frac{1}{n^3})$.

Let $\cL$ be the set of lines through $q$ that contain a grid point. This is a set of $O(n^6)$ lines. Consider now a shortest \TSPN tour of $\cL$ that is disjoint from $B(q,\frac{1}{n^3})$. The sets $\ell\setminus B(q,\frac{1}{n^3})$ for each $\ell \in \cL$ have pairwise distance at least $\frac{1}{80n^6}$, so the tour must have length at least $40n^3 \cdot  40n^3 \cdot \frac{1}{80n^6} = 20$.
\end{proof}

\recallthm{lem:noptas_main}

\begin{proof}%[Proof of \Cref{lem:noptas_main}]
To prove the first claim, let $W$ be a a vertex cover of $G$ of size
$k$. We can create a tour $T$ by first adding all edges of $\gamma$,
and for each vertex $w\in W$, we add a detour: if $q,q'$ are the
nearest points of $Q$ to $w$, then we remove the segment $qq'$ of
length $\delta$ and add the segments $qw$ and $wq'$ of length
$10\delta$ each to the tour. For each vertex $w\in W$, this results in
a length increase of $19\delta$, so the resulting tour $T$ has length
$10+19\delta k$ as required. We can verify that $T$ touches every line of $\cL^*$. It goes through each $q\in Q$, thus it goes through all lines in point gadgets. For each line $\ell\in \bar{L}$ the corresponding graph edge $uv$ is covered by the vertex cover, so either $u\in W$ or $v\in W$. Therefore $\ell$ is touched either at $Ap(u)$ or $Ap(v)$.

To prove the second claim, let $T$ be a tour of length $10+19\delta k$. Since $\delta=1/(4000n)$ and $k\leq n$, we have that the length of $T$ is less than $20$. Since $T$ touches each line, it is also a valid tour for any subset of lines. In particular, for each $q\in Q$ it is a tour of length less than $20$ of the point gadget of $q$. Consequently, $T$ intersects each ball $B(q,1/n^3) (q\in Q)$. Note that by the properties of $Q$, these balls are disjoint and have pairwise distance more than $2/n^3$ if $n$ is large enough. Let $B_q$ denote the ball $B(q,1/n^3)$.

Without loss of generality, we can assume that $T$ is a $3$-dimensional (skew) polygon whose vertices are on the lines of $\cL^*$. Consider the vertices of $T$ in order, and remove all vertices of the sequence that are only incident to lines of point gadgets, but lie outside the balls $B_q (q\in Q)$. Furthermore, remove entries that fall inside $\bigcup_{q\in Q} B_q$ until we get a sequence where there is a unique vertex $h_i$ from each $B_q$.
Let $h=(h_1,\dots,h_m)$ be the sequence of vertices we get this way.\footnote{The sequence $h$ should be understood as a cyclic sequence, where indices are defined modulo $m$. In particular $h_{m+1}=h_1$.} As a result, for each $\ell\in \bar{\cL}$ there exists a point $h_i\in \ell$, and for each $q\in Q$, there is some unique entry $h_{j}\in B_q$. Fix an orientation of $T$, and let $T(h_i,h_j)$ denote the subpath of $T$ from $h_i$ to $h_j$. The balls $B_q$ partition $T$ into $|Q|$ subsequences, so $h$ can be regarded as the concatenation of sequences $g^1,g^2,\dots,g^{|Q|}$ where for each $j\in \{1,\dots, |Q|\}$ we have $g^j=(g^j_0,g^j_1,\dots,g^j_{t(j)})$, $g^j_0\in \bigcup_{q\in Q} B_q$, and for each $i\in \{1,\dots,t(j)\}$ it holds that $g^j_i \not\in \bigcup_{q\in Q} B_q$.

\begin{claim}\label{cl:shortstep}
$\cost(T(g^j_0,g^{j+1}_0))\geq \delta - 2/n^3$.
\end{claim}

\begin{proof}
Let $q,q'$ be points of $Q$ such that $g^j_0\in B_q=B(q,1/n^3)$ and $g^{j+1}_0\in B_{q'}=B(q',1/n^3)$. By the definition of $Q$, we have $\dist(q,q')\geq \delta$. Consequently,
$\cost(T(g^j_0,g^{j+1}_0))\geq \dist (g^j_0,g^{j+1}_0) \geq \delta - 2/n^3$.
\end{proof}

Consider now a sequence $g^j$, and let $u^j_1,v^j_1,u^j_2,v^j_2, \dots,u^j_{t(j)},v^j_{t(j)}\in \bar{P}$ be a sequence of points such that the point $g^j_i$ is on the line $\ell(u^j_i,v^j_i)$.

\begin{lemma}\label{lem:cylinder}
There do not exist lines $\ell_1,\ell_2,\ell_3\in \{\ell(u_1v_1), \dots \ell(u^j_{t(j)},v^j_{t(j)})\}$ such that $\ell_1$ connects $\bar{e}^1$ with $\bar{e}^2$, line $\ell_2$ connects $\bar{e}^2$ with $\bar{e}^3$, and $\ell_3$ connects $\bar{e}^3$ with $\bar{e}^1$.
\end{lemma}

\begin{proof}
First, we show that $T$ is disjoint from the cylinder $\cY$ with axis $(0,0,0)^T(1,1,1)^T$ and radius $\sigma/2$. The cylinder $\cY$ has distance more than $\sigma/3$ from the points of $Q$, so a tour $g^j$ touching the cylinder has length at least $2\sigma/3$. By \Cref{cl:shortstep}, we get \[\cost(T)= \sum_{j=1}^{|Q|} \cost(T(g^j_0,g^{j+1}_0))\geq (|Q|-1)(\delta - 2/n^3) + 2\sigma/3 = 10 + 2\sigma/3 - \delta - O(1/n^2)>10.5,\] which is a contradiction as $\cost(T)<10+20n\delta<10.5$.

Suppose for the sake of contradiction that $\ell_1,\ell_2,\ell_3$ are lines touched by $T$ between touching $B_q$ and $B_{q'}$. In particular, the portion of the tour between $B_q$ and $B_{q'}$ contains a path $T'$ that is disjoint from $\cY$, and goes from $\ell_1$ to $\ell_3$, but touches $\ell_2$ on the way. Let us project the lines and the tour into the plane $H:x+y+z=0$ perpendicularly, and denote the projection with $\pi(.)$. We have $\cost(\pi(T'))\leq \cost(T')$, and since $\cY$ is perpendicular to $H$, the path $\pi(T')$ is disjoint from the disk $\pi(\cY)$, which is a disk of radius $\sigma/2$ in $H$ centered at the origin. Notice that $\pi(\bar{e}^a)$ form the three non-adjacent sides of a regular hexagon in $H$ centered at the origin, see Figure~\ref{fig:cylinder}(i).

\begin{figure}
\centering
\includegraphics[scale=0.9]{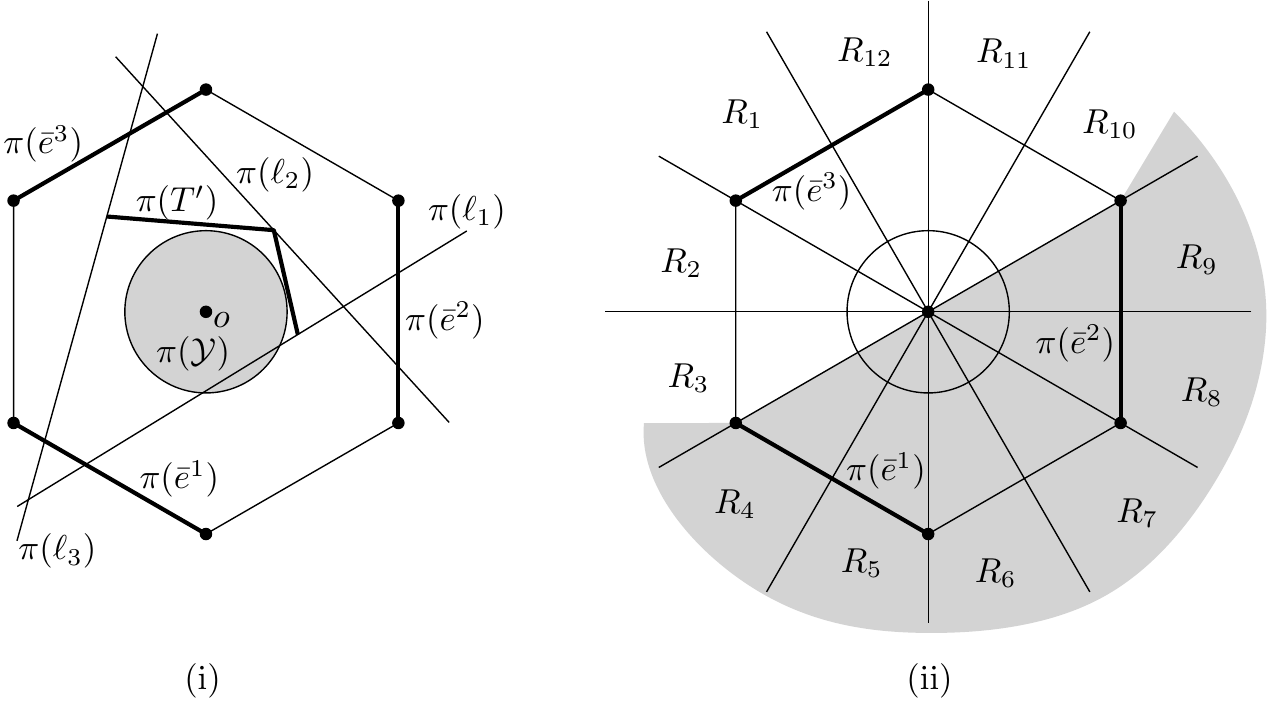}
\caption{(i) Projection into the plane $H:x+y+z=0$. (ii) Twelve cones, eight of which covers the line $\pi(\ell_1)$.}\label{fig:cylinder}
\end{figure}

Let $R_1,\dots,R_{12}$ be the twelve cones centered at the origin whose boundary contain a midpoint and an endpoint of a hexagon side in cyclic order. We can define the indices so that $\pi(\bar{e}^a)\subset R_{4a}\cup R_{4a+1}$, see Figure~\ref{fig:cylinder}(ii). For any choice of $\ell_1$ we have that $\pi(\ell_1)$ intersects both $\pi(\bar{e}^1)$ and $\pi(\bar{e}^2)$, therefore it is covered by the eight regions $R_3\cup R_4\cup\dots\cup R_{10}$. In general $\pi(\ell_j)$ is covered by $R_{4j-1}\cup R_{4j}\cup\dots\cup R_{4j+6}$. We claim that $\pi(T')$ intersects at least two non-adjacent regions among $R_i\setminus \pi(\cY)\, (i=1,\dots,12)$. Indeed, if $\pi(T')\subset R_i\cup R_{i+1}\setminus \pi(\cY)$ for some $i\in \{1,\dots,12\}$, then $\pi(T')$ is disjoint from at least one of $R_{4j-1}\cup R_{4j}\cup\dots\cup R_{4j+6}\, (j=1,2,3)$, so it cannot touch $\pi(\ell_j)$. Consequently, $T'$ does not touch $\ell_j$, which contradicts the definition of $T'$. Therefore $\pi(T')$ intersects at least two non-adjacent regions.

The distance of two non-adjacent regions $R_i\setminus \pi(\cY)$ and $R_j\setminus \pi(\cY)$ is at least $\sigma/4$, which implies that $\cost(T')>\sigma/4$. Similarly to the calculation seen above for the cylinder, we get $\cost(T)>10 + \sigma/4 - o(1)$, which is a contradiction.
\end{proof}

By \Cref{lem:cylinder}, we may assume without loss of generality that $u^j_1,\dots,u^j_{t(j)}\in \bar{e}^1$. Note that whenever $u^j_i v^j_i$ and $u^j_{i+1}v^j_{i+1}$ correspond to non-incident edges, (equivalently, when $u^j_i\neq u^j_{i+1}$) then we have that $\dist(\ell(u^j_i,v^j_i),\ell(u^j_{i+1},v^j_{i+1}))>1/(200n)$ by \Cref{cor:mindist}, therefore $\cost(T(g^j_i,g^j_{i+1}))> 1/(200n)=20\delta$. If there are $s$ unique points in the sequence $u_1,\dots,u_t$, then $\cost(T(g^j_1,g^j_{t(j)}))> 20(s-1)\delta$.  We add each $u^j_i$ into a set $W$. We execute the same procedure on each sequence $g^j$. We claim that the resulting set $W$ is a vertex cover of size at most $1.011k$.

The set $W$ is a vertex cover as each line $\bar{\ell}(vw)$ is visited by $T$, therefore $v$ or $w$ appears in the sequence $u^j_1,\dots,u^j_{t(j)}$ for some subinterval $g^j$ of $h$ and therefore $v$ or $w$ gets added to $W$.

It remains to prove the bound on the size of $W$.
We give a lower bound on $\cost(T(g^j_0,g^{j+1}_0))$. If the sequence $g^j$ has a single entry (that is, $t(j)=0$), then $\cost(T(g^j_0,g^{j+1}_0)\geq \delta -2/n^3$ by \Cref{cl:shortstep}. Otherwise, by \Cref{cl:qdist} we have that $T(g^j_0,g^j_1)$ and $T(g^j_{t(j)},g^{j+1}_0)$ both has cost at least $9.9\delta$, and if there are $s_j$ unique vertices contributed from $u^j$ to $W$, then by the arguments above $\cost(T(g^j_1,g^j_{t(j)}))> 20(s_j-1)\delta$. Therefore $\cost(T(g^j_0,g^{j+1}_0)) > 20(s_j-1)\delta+ 2\cdot9.9\delta \geq 18.8s_j\delta + \delta$. Putting the case $s_j=0$ and $s_j>0$ together, we get that $\cost(T(g^j_0,g^{j+1}_0))> 18.8s_j\delta + \delta -2/n^3$.
Consequently, $\cost(T)>\sum_{j=1}^{|Q|}(18.8s_j\delta +\delta-2/n^3) \geq 10 + |W|\cdot18.8\delta-O(1/n^2)$. Since $\cost(T)=10+19\delta k$, we have that $|W|\leq \frac{19+O(\frac{1}{\delta n^2})}{18.8}k=\frac{19+O(1/n)}{18.8}k$, so for $n$ large enough we have $|W|<1.011k$.
\end{proof}

\section{Detailed Proofs from Section~\ref{sec:hardness-line}}

\recallthm{lem:distance_point_line}

\begin{proof}%[Proof of \Cref{lem:distance_point_line}]
  We distinguish two cases: (i) when $\ell$ is incident to $\ell_{uw}$
  and (ii) when $\ell$ is not incident to $\ell_{uw}$.

  For the first case, assume without loss of generality that $\ell = \ell_{uz}$. Then, the distance
  between $p$ and $\ell$ is given by
  \begin{align*}
    \lvert up\rvert\cdot \sin(\an(\ell,\ell_{uw})).
  \end{align*}
  Since  $\lvert up\rvert \ge \Delta$, it suffices to show that
  $\sin(\an(\ell,\ell_{uw}))\ge 1/2\Leftrightarrow \an(\ell,\ell_{uw})
  \ge \pi/6$. 

  Consider the ``almost equilateral'' triangle $uwz$. By applying the
  law of cosines, we have
  \begin{align*}
    \lvert wz \rvert^2 &= \lvert uw \rvert^2 + \lvert uz \rvert^2 -
    2\lvert uw \rvert\lvert uz \rvert\cos \an(\ell,\ell_{uw})\\
    &\Rightarrow
    \cos \an(\ell,\ell_{uw}) \le \frac{2(1+\delta)^2-1}{2}.
  \end{align*}
  We note that for $\delta\le \sqrt{(\sqrt{3}+1)/2} - 1$ we have
  \begin{align*}
    \sin\left(\an(\ell,\ell_{uw})\right) =
    \sin\left(\arccos\left((1+\delta)^2-\frac12\right)\right) \ge
    \sin\left(\arccos\left(\frac{\sqrt{3}}{2}\right)\right) = \frac{1}{2}
  \end{align*}

  For the second case, assume that $\ell = \ell_{zv}$. We note that
  $zv$ and $uw$ are non-incident edges of an ``almost regular''
  tetrahedron (formed by the points $\{z,v,u,w\}$). Since the distance of non-incident edges of a regular
  tetrahedron of edge-length $1$ is $1/\sqrt{2}$ and the edge lengths
  in our tetrahedron are in the range $[1,1+\delta]$, there exists
  a small enough $\delta$ (which is independent of $\epsilon$) such
  that $\dist(p,\ell)\ge 1/2 > \Delta/2$.

\end{proof}

\recallthm{lem:covering_at_balls}

\begin{proof}%[Proof of Lemma~\ref{lem:covering_at_balls}]
  Note it is without loss of generality to assume that $T$ consists of
  line segments with endpoints on lines of $I(G')$. By
  Lemma~\ref{lem:distance_point_line} any line
  $\ell_{u^iw^j}\in\mathcal{L}_{uw}$ that is visited at a point $p$
  with $p\not\in B(u^i)$ and $p\not\in B(w^j)$, must have two adjacent
  segments on $T$ of length at least $\Delta/2$ each. Since the total
  tour cost  of $T$ is at most $x$ there
  can be at most 
 % \begin{align*}
  $  \frac{x}{\Delta/2} =
    \frac{2 x}{\Delta}$
 % \end{align*}
  lines that are visited by $T$ outside a ball.
\end{proof}

\recallthm{lemma:non-empty-balls}

\begin{proof}%[Proof of Lemma~\ref{lemma:non-empty-balls}]
We first argue that $S$ is a vertex cover of $G$. Assume for
the sake of contradiction that some edge $uv\in E(G)$ is not covered
by $S$. Then it must be the case that there are at least $\alpha(1-\lambda)$ empty balls among the balls corresponding to both $u$ and $v$. But any
line defined by two such empty balls corresponding to $u$ and $v$ 
is not covered by a ball. In total there are more than
$(1-\lambda)^2\alpha^2=\Omega(n^4)$ many
such lines. This is a contradiction,
since by Lemma~\ref{lem:covering_at_balls} there can be at most
$\frac{2n \alpha k (1+\delta)}{\Delta}=O(n^3)$ such lines in total over the
whole instance.

Let $S$ be the vertex cover of $G$ we have obtained. Since the dsitance of any two balls is at least $1-2\Delta$, and we have visited at least $\alpha\lambda$ balls among $Q_v$ for each $v\in S$,
the total cost of the tour is at least
\[x>|S|\alpha\lambda(1-2\Delta),\]
therefore we have that $|S|<\frac{x}{\alpha \lambda (1-2\Delta)}$.
\end{proof}

%!TEX root=paper.tex

\section{Details of Section~\ref{sec:super}}

\subsection{Arora's Algorithm}
\label{sec:app:super:arora}

Arora's algorithm consists of three main steps:
\begin{enumerate}[noitemsep]
\item Perturbation, which changes the instance so that all coordinates are integral and bounded by $O(n)$;
\item Construction of a shifted quadtree;
\item Dynamic program, which finds the approximate solution for TSP.
\end{enumerate}

We describe all of these steps, including any minor alterations needed for them to work in our setting.

\subsubsection{Perturbation}
\label{sec:super:arora:1}

 Arora shows how to perturb the solution such that:
\begin{enumerate}[noitemsep]
  \item All nodes have integer coordinates;
  \item Every (non-zero) distance between two points is at least $8$ units;
  \item The maximum distance between two points is $O(n)$.
\end{enumerate}

Given a bounding box on the instance of size $L_0$, Arora achieves this
perturbation by snapping points to an appropriately fine grid. %
To use this step for our problem, we need to specify a value of $L_0$ such that 
${\OPT \leq L_0 \leq O(\OPT)}$. %
To this effect, we guess the value of $\OPT$ rounded up to a power of $2$, as
well as a vertex $v_0$ that is included in an optimum solution. 
We implement this guessing step by iterating over all of the possible values, and computing a feasible solution for each possibility. %
The best feasible solution we obtain will be at least as good as the solution
for the correct guess (in expectation).

The guessing step is done as follows. 
We start by {guessing} a vertex $v_0$ that is contained in an optimum
solution. %
Then, we compute the minimum radius $R_0$ such that at least one point from each
neighborhood is contained in the ball $B$ of radius $R_0$ centered at $v_0$. %
Such a ball can be computed simply by iterating over all neighborhoods and
finding the neighborhood's nearest point to $v_0$. %
If the optimum solution contains $v_0$, then its cost is at least $R_0$, as it
must visit the farthest neighborhood, at distance $R_0$. %
On the other hand, $\OPT \leq 2R_0n$, since the ball $B$ contains at least one
point from each neighborhood, and the distance between any two points in $B$
is at most $2R_0$. %
Hence, there is a tour of cost at most $2R_0n$. %
Knowing that $R_0 \leq \OPT \leq 2R_0n$ (assuming $v_0$ is in an optimum solution),
we can simply run the algorithm for every $v_0$ and for any $R \in
[R_0, 4R_0n]$ that is a power of $2$. %

Given a vertex $v_0$ and a guess $R$ for the value of the optimum solution, we
set $L_0 = R/2$ (so that if $R/2 \leq \OPT \leq R$, $L_0 \leq \OPT$). %
Finally, we remove all of the vertices $u \in P$ that are at a distance more
than $R$ from $v_0$, that is, $\dist(v_0, u) > R$. %
A solution containing both $v_0$ and $u$ would cost more than $R\geq \OPT$,
implying that for correct choices of $R$ and $v_0$, such vertices can be safely
removed. We now have a bounding box of {side length} $4L_0$ {containing} all the points in
the instance, and hence the perturbation step in Arora's algorithm ensures the
stated properties.

\subsubsection{Construction of a shifted quadtree}
\label{sec:super:arora:2}

Let $L = O(n)$ be the size of the bounding box. %
The algorithm computes a random shift $a = (a_1, a_2, \ldots, a_d)$, with $a_i \in \{0,
\ldots, L-1\}$, $i \in [d]$. %
Then, it constructs a quadtree where the dissection points are shifted
according to $a$. %
The resulting quadtree has height $O(\log n)$, and $O(n\log n)$ cells.
For our purpose, no changes are needed to this process.

\subsubsection{Dynamic Program}
\label{sec:super:arora:3}

Arora's algorithm uses dynamic programming to find a \emph{salesman path},
which may visit additional points along the boundary of the cells of the
quadtree. The following definition formalizes this concept.

\begin{definition}
Let $m$, $r$ be positive integers. An \emph{$m$-regular set of portals} for a
shifted dissection is a set of points on the facets of the cells in it. Each
cell has a portal at each of its vertices and $m$ other portals on each facet,
placed in a $d-1$-dimensional square grid whose vertices are identical to the vertices of the facet.

A \emph{salesman path} is a path in $\Reals^d$ that visits all the input
points, and some subset of portals. %
It may visit a portal more than once.

The salesman path is \emph{$(m, r)$-light} with respect to the shifted dissection if it
crosses each facet of each cell in the dissection at most $r$ times and always at
a portal.
\end{definition}

The goal of the dynamic program is to find a minimum cost $(m,r)$-light
salesman path, for the instance. %
For our purpose, a $2$-approximation of \TSP is sufficient, and hence we set 
${m = O(\sqrt d \log n)^{d-1}}$ and ${r = O(\sqrt d)^{d-1}}$. %
By restricting the solution to cross the cell boundaries only through portals,
we can see that any solution to the problem, restricted to a single cell,
consists of a set of paths that together cover all of the points inside the
cell. %
Since we want to find an $(m,r)$-light solution, this further implies that at
most $r$ portals per facet of the cell are used. %
This motivates the definition of the \emph{$(m,r)$-multipath problem}, which is the problem solved by the dynamic program for each cell:

\begin{definition}[$(m,r)$-multipath problem \cite{Arora98}]
\label{superpoly:multipath}
An instance of this problem is specified by the following inputs:
\begin{enumerate}[noitemsep]
  \item A nonempty cell in the quadtree.
  \item A multiset of $r$ portals on each of the $2d$ facets of this cell
  such that the sum of the sizes of these multisets is an even number
  $2p\leq2dr$. \label{superpoly:multipath:2}
  \item A pairing $(a_1, a_2), (a_3, a_4), \ldots (a_{2p-1}, a_{2p})$
  between the $2p$ portals specified in \Cref{superpoly:multipath:2}.
\end{enumerate}
The goal in the $(m, r)$-multipath problem is to find a minimum cost
collection of $p$ paths in the cell that is $(m, r)$-light. %
The $i$-th path connects $a_{2i-1}$ to $a_{2i}$, and
the $p$ paths together visit all the points in the cell.
\end{definition}

The dynamic programming table consists of all of these instances of
$(m,r)$-multipath problem, for each cell and pairing of portals %
(considered here to include the multiset of portals in Item
\ref{superpoly:multipath:2}. %
We refer to the multiset of portals and their pairing as the \emph{state} of an $(m,r)$-multipath problem.

The values of the table can be computed recursively. %
The entries corresponding to leaves of the quadtree can be easily determined:
given the portal set of size $2p$ and the pairing, we simply need to find the
shortest paths between the paired portals, and add the (single) point in the
cell to one of these paths. %
For all other entries, the algorithm enumerates all possible ways that the $p$
paths can cross the boundary between children cells. % 
For each of these arrangements, the cost of the solution can be obtained by
summing the costs of the respective instances for the children cells. %
Once all of the entries have been computed, the minimum cost $(m,r)$-light
salesman path can be found by looking at the $(m,r)$-multipath problem for the
root cell of the quadtree with no portals used. %

The dynamic programming table contains a total of $O\paren[\big]{n (\log
n)^{O({d})^{(d-1)/2}}}$ entries, %
and the value at each cell can be computed in time $(\log n)^{O({d})^{(d-1)/2}}$. %
Therefore, the running time of this algorithm is $O\paren[\big]{n (\log n)^{O({d})^{(d-1)/2}}}$.

Our algorithm uses a very similar dynamic program, with only a small change needed at the leaves. %
In the \TSP problem, all of the points must be visited, which implies that any
feasible solution to the \mbox{$(m,r)$-multipath} problem must visit all the points
contained in that cell. %
However, the same is not true of the \TSPN problem: %
as long as one point from each neighborhood is visited in the whole path, the
solution is feasible, which means that not all neighborhoods are visited in
every cell that intersects them. %
To that effect, we add an extra input to the $(m,r)$-multipath problem for
leaf cells, which we call \emph{visit bit}. %
If the visit bit is set to \strue, then the (single) point in the cell
must be visited; if it is set to \sfalse, then the solution only needs
to connect the portals as specified in the input (meaning that the optimum
solution will be a union of shortest paths between paired portals).

\subsection{Approximating \TSPN using the framework by Chalermsook et al.}
\label{sec:app:super:approx}

To prove \Cref{thm:super:thm}, we need to show how to formulate \TSPN as an
instance of \STGST, and then show how to use \Cref{thm:super:approx:round} to
obtain an $O(\log^2 n)$-approximation. %
In this section, we provide details to both of these steps.

\subsubsection{Formulating discrete \TSPN as an instance of \STGST}
\label{sec:super:approx:1}

We will formally describe the construction of a  DAG $H$ based on the dynamic
program for \TSP. %
We assume that the perturbation and random shift steps implemented by Arora
have been performed, with the alterations described in \Cref{sec:super:arora}.

We now consider the dynamic program as presented by Arora, and construct
our DAG $H$ as follows. %
The vertex set is partitioned into \emph{subproblem nodes} $H_p$ and \emph{combination nodes} $H_c$.
\begin{itemize}[noitemsep]
  \item For every $(m,r)$-multipath subpro{}blem considered by Arora, we create a \emph{subproblem node}. %
  Formally, for every cell $C$ in the quadtree, and every state $A$ or
  $(A,b)$ (where $b$ represents the visit bit if $C$ is a leaf cell), we
  create a node $t[C, A]$ (resp. $t[C,A,b]$) in $H_p$.
  \item For every non-leaf cell $C$ with children $C_1, \ldots, C_k$ and states $X$ for $C$ and $X_i$ for $C_i$, we add a \emph{combination node} $t^c[C, X, \{X_i\}_{i \in [k]}]$ if the states are consistent, that is, if the combination of the portal pairings for each of the cells $C_i$ forms the portal pairing represented by $X$ in $C$.
  \item For each combination node $t' = t^c[C, X, \{X_i\}_{i \in [k]}]$, we add edges from $t[C,X]$ to $t'$ and from $t'$ to $t[C_i, X_i]$ for each $i \in [k]$.
  \item The edges entering leaf nodes $t[C,A,b]$ have cost equal to the minimum cost of a solution to the $(m,r)$-multipath problem in $C$ with portal pairings specified by $A$, and which visits the point in $C$ if $b = \strue$. 
  \item All other edges have cost $0$.
\end{itemize}
The root of $H$ is the node $t[C,X]$, where $C$ is the root cell of the quadtree (the bounding box of the instance), and $X$ represents an empty set of portals. %

\begin{lemma}
\label{lem:super:approx:equiv}
Let $v_0 \in P$ be a point and $R_0$ be a radius {guessed} in
\Cref{sec:super:arora:1}.

For every $(m,r)$-light tour $F$ in the resulting quadtree
{there is a solution tree $X$ in $H$ such that $\cost(F)=\cost(X)$ and they visit the same set of points in $P$.}

Similarly, for any solution tree $X \subset H$, there is an $(m,r)$-light tour
$F$ of the same cost, which visits the same points in $P$.
\end{lemma}

\begin{proof}
Given an $(m,r)$-light tour $F$, we can define a solution tree $X$ by choosing
its subproblem nodes. %
We show that if the subproblem nodes are chosen consistently,
then the corresponding combination nodes and the edges connecting all of the subproblem nodes exist in the graph, and thus we obtain a solution tree. %
For every non-leaf cell $C$, we choose the subproblem $t[C,A]$ such that $A$
describes the portals used by $F$ on the boundary of $C$, as well as how they
are pairing. %
For every leaf cell $C$, we choose the subproblem $t[C,A,b]$ such that $A$ is
as described above and $b$ is set to $\strue$ if the point in $C$ is visited
by $F$ and $\sfalse$ otherwise. %
By definition, the portal pairings of a non-leaf cell $C$ are consistent with
the portal pairings of the children cells (as they describe the portals used
by $F$). %
Therefore, there is a combination node for every node $t[C,A]$ above,
connecting it to the nodes corresponding to the children cells.

As to the cost, notice that the cost of $F$ equals the sum of the costs of $F$
on each of the leaf cells. %
Similarly, the cost of a solution tree $X$ is the sum of the costs of the
edges entering leaf nodes $t[C,A,b]$. %
Since $F$ restricted to a leaf cell $C$ is a solution to the $(m,r)$-multipath
subproblem encoded by $t[C,A,b]$, the cost of $F$ restricted to $C$ must be at
least as much as the optimum to that subproblem, which is the cost of the edge
entering $t[C,A,b]$. %
Taking into account that all other edges have cost $0$, and that each leaf
node $t[C,A,b]$ of the solution tree corresponds to a part of $F$ whose cost
is at least as much as the cost of the edge entering $t[C,A,b]$, we conclude
that $\cost(X) \leq \cost(F)$.

For the second part of the proof, let $X$ be a given solution tree. %
We construct $F$ by taking the union of the optimum solutions to the
$(m,r)$-multipath subproblems corresponding to the leaf nodes $t[C, A, b]$ of
solution tree $X$. %
The cost of $X$ is the sum of the costs of edges entering such leaf nodes,
each of which is the optimum cost of the corresponding $(m,r)$-multipath
subproblem, which is the cost of that part of $F$. %
Since all cell leaves are disjoint, we conclude that $\cost(F) = \cost(X)$. %

To complete the proof, we need to show that $F$ is a circuit. %
We start by showing that any path between two portals must have a
continuation, that is, the number of paths incident on each portal is even. %
Then we show that the solution must be connected, and thus forms a single
circuit. %
Let $q$ be a portal contained in $F$, and $C$ be the smallest cell such that $q$ is contained
in $C$ but is not one of its portals. %
$C$ must exist, since no portal of the cell corresponding to the bounding box
can be used. %
By minimality of $C$, two or more children cells of $C$ contain $q$. %
As part of the recursion rules of Arora's dynamic program, each solution to an
$(m,r)$-multipath subproblem must contain paths between portals of $C$ or a
single circuit. %
Therefore, the degree at $q$ must be even. %

Similarly, the solution cannot be the disjoint union of multiple circuits. %
Assume otherwise. %
Then, there is a smallest cell $C$ containing two circuits or a circuit and
some other paths. %
By minimality of $C$, either one of the children cells contains a circuit of
$F$, and the recursion rules of Arora's dynamic program prevent any other
child cell from containing a part of the solution, or the portal pairings
themselves induce two circuits or a circuit and some other paths, which would
not be permitted by Arora's algorithm.
\end{proof}

\subsubsection{Obtaining an \texorpdfstring{$O(\log^2n)$}{O(log2 n)}-approximation}

We will now show how to use \Cref{thm:super:approx:round} and
\Cref{lem:super:approx:equiv} to obtain an $O(\log N \log n)$-approximation
for the \TSPN problem on discrete neighborhoods, and hence prove
\Cref{thm:super:thm}.

We start by guessing a vertex $v_0$ to be the starting point of our solution. %
For every vertex $v_0 \in P$, we compute the minimum radius $R_0$ such that every neighborhood contains a point at distance at most $R_0$ from $v_0$. %
Next, we guess $R$, an approximation for $\OPT$, in the range $[R_0, 4nR_0]$. %
For the powers ${R=2^i}$, $i\in \mathbb{Z}, R_0\leq R \leq 4nR_0$, we can now preprocess the
instance according to the perturbation step of Arora's algorithm.
(\Cref{sec:super:arora:1}). %
Next, we enumerate the shift $a = (a_1, \ldots, a_d) \in \set{0, \ldots,
L-1}^d$, and construct the shifted tree as in Arora's algorithm
(\Cref{sec:super:arora:2}). %
Finally, we construct the DAG $H$ based on the dynamic programming table, as
specified in \Cref{sec:super:approx:1}. %
We recall that the height of the tree, as well as of DAG $H$ is $O(\log N)$.

We now use \Cref{thm:super:approx:round} repeatedly to obtain solution trees
$X_1, \ldots, X_\ell$, where $\ell = c \log n \log N$, and $c$ is a large
constant. %
Then, we use \Cref{lem:super:approx:equiv} to convert each solution tree $X_i$
into a tour $F_i$, and finally take the union of all these tours to obtain a
solution $F$. %
While $F$ is not necessarily a tour, it is simple enough to remove crossings. %
For every neighborhood $P_i$ that is not visited by $F$, we add a detour
visiting the closest point in $P_i$. %
We denote by $F^*$ the minimum-cost solution among all solutions $F$ for all
the enumerated values of $v_0$, $R_0$, and $a$.

By construction, $F^*$ is a feasible solution, as it is a tour that visits every group. %
To prove that it is $O(\log N \log n)$-approximate, consider the solution $F'$ that we
obtained for the correct values of $v_0$, $R_0$, and $a$, that is, for a vertex
$v_0$ in an optimum solution, $R_0$ such that $R_0/2 \leq \OPT \leq R_0$, and a shift
$a$ for which an $(m,r)$-light tour exists. %
By \Cref{thm:super:approx:round}, each of the solution trees $X'_i$ obtained
has expected cost at most $\OPT$, and by \Cref{lem:super:approx:equiv}, the
corresponding tour $F'_i$ also has expected cost at most $\OPT$. %
Therefore, the union of all tours $F'_i$ costs at most $O(\log N \log n
\OPT)$ in expectation. %
The probability that a neighborhood is not visited, and hence that we must add
a detour, is (for sufficiently large {$c$})
\begin{align*}
\Prob\brac*{\bigcap_j |S_i \cap X_j| = 0}
&\leq \left(1-\frac1{\alpha\height(H)}\right)^\ell \\
&\leq e^{-O(\log n)} \\
&\leq \frac1{n^3}
\end{align*} 
We conclude that the expected cost of $F'$ is at most $O(\log N \log n
\OPT)$, and since, by \Cref{lem:super:approx:equiv}, ${\cost(F^*) \leq \cost(F')}$, $F^*$ is $O(\log N \log n)$-approximate in expectation. %
By \Cref{thm:super:approx:round}, the running time of our algorithm is %
\[N^{O(d)}\, (\log N)^{O(d)^{(d-1)/2)}\, O(\log N)}= N^{O(d)^{(d-1)/2} \log \log N}.\] %

This completes the proof of \Cref{thm:super:thm}.

\end{document}